\documentclass[letterpaper,11pt]{article}

\usepackage[margin=1in]{geometry}  
\usepackage{xspace,exscale,relsize}
\usepackage{fancybox,shadow}
\usepackage{graphicx}
\usepackage{color}
\usepackage{amsthm,amsfonts}
\usepackage{amssymb}
\usepackage{authblk}
\usepackage[title]{appendix}

\usepackage{makecell}

\usepackage{microtype}

\usepackage[
            CJKbookmarks=true,
            bookmarksnumbered=true,
            bookmarksopen=true,
            colorlinks=true,
            citecolor=red,
            linkcolor=blue,
            anchorcolor=red,
            urlcolor=blue,
            pdfauthor={Polyanskiy and Wu}
            ]{hyperref}

\usepackage[noadjust]{cite}

\usepackage{amsmath}
	\makeatletter
	\let\over=\@@over \let\overwithdelims=\@@overwithdelims
	\let\atop=\@@atop \let\atopwithdelims=\@@atopwithdelims
  	\let\above=\@@above \let\abovewithdelims=\@@abovewithdelims
  	\makeatother
\usepackage{fixltx2e}

\usepackage{rotating}

\usepackage{ifpdf}

\usepackage{subfigure}
\usepackage{psfrag}

\usepackage{prettyref,enumerate}

\usepackage[all]{xy}

\usepackage{tikz}
\usetikzlibrary{arrows}
\tikzstyle{int}=[draw, fill=blue!20, minimum size=2em]
\tikzstyle{dot}=[circle, draw, fill=blue!20, minimum size=2em]
\tikzstyle{init} = [pin edge={to-,thin,black}]

\usepackage{tikz-cd}

\newcommand{\matc}{\ensuremath{\mathcal{C}}}
\newcommand{\matx}{\ensuremath{\mathcal{X}}}

\newcommand{\maty}{\ensuremath{\mathcal{Y}}}

\newcommand{\matp}{\ensuremath{\mathcal{P}}}

\newcommand{\mreals}{\ensuremath{\mathbb{R}}}

\newcommand{\pa}{\ensuremath{\mathrm{pa}}}

\ifx\eqref\undefined
	\newcommand{\eqref}[1]{~(\ref{#1})}
\fi
\ifx\mod\undefined
	\def\mod{\mathop{\rm mod}}
\fi

\def\EE{\Expect}
\def\FF{\mathbb{F}}

\def\PP{\mathbb{P}}

\def\eqdef{\triangleq}

\def\perc{\mathrm{perc}}

\newcommand{\BSC}{\mathsf{BSC}}
\newcommand{\BEC}{\mathsf{BEC}}
\newcommand{\EC}{\mathsf{EC}}
\newcommand{\Binom}{\mathrm{Binom}}

\newcommand{\LC}{\mathrm{LC}}

\newcommand{\etaKL}{\eta_{\rm KL}}
\newcommand{\etachi}{\eta_{\chi^2}}
\newcommand{\etaTV}{\eta_{\rm TV}}

\newcommand{\reals}{\mathbb{R}}

\newcommand{\Expect}{\mathbb{E}}

\newcommand{\TV}{d_{\rm TV}}

\newcommand{\expects}[2]{\mathbb{E}_{#2}\left[ #1 \right]}
\newcommand{\diff}{{\rm d}}
\newcommand{\eg}{e.g.\xspace}
\newcommand{\ie}{i.e.\xspace}
\newcommand{\iid}{i.i.d.\xspace}

\newcommand{\fracd}[2]{\frac{\diff #1}{\diff #2}}

\newcommand{\pth}[1]{\left( #1 \right)}
\newcommand{\qth}[1]{\left[ #1 \right]}
\newcommand{\sth}[1]{\left\{ #1 \right\}}

\newcommand{\iiddistr}{{\stackrel{\text{\iid}}{\sim}}}

\newcommand\indep{\protect\mathpalette{\protect\independenT}{\perp}}
\def\independenT#1#2{\mathrel{\rlap{$#1#2$}\mkern2mu{#1#2}}}
\newcommand{\Bern}{\text{Bern}}

\newcommand{\indc}[1]{{\mathbf{1}_{\left\{{#1}\right\}}}}
\newcommand{\Indc}{\mathbf{1}}

\definecolor{myblue}{rgb}{.8, .8, 1}
\definecolor{mathblue}{rgb}{0.2472, 0.24, 0.6} 
\definecolor{mathred}{rgb}{0.6, 0.24, 0.442893}
\definecolor{mathyellow}{rgb}{0.6, 0.547014, 0.24}

\newcommand{\tx}{{\tilde{x}}}

\newcommand{\calL}{{\mathcal{L}}}

\newcommand{\calV}{{\mathcal{V}}}

\newcommand{\calX}{{\mathcal{X}}}
\newcommand{\calY}{{\mathcal{Y}}}

\newcommand{\diverge}{\to \infty}

\newrefformat{eq}{(\ref{#1})}
\newrefformat{thm}{Theorem~\ref{#1}}
\newrefformat{th}{Theorem~\ref{#1}}
\newrefformat{chap}{Chapter~\ref{#1}}
\newrefformat{sec}{Section~\ref{#1}}
\newrefformat{algo}{Algorithm~\ref{#1}}
\newrefformat{fig}{Fig.~\ref{#1}}
\newrefformat{tab}{Table~\ref{#1}}
\newrefformat{rmk}{Remark~\ref{#1}}
\newrefformat{clm}{Claim~\ref{#1}}
\newrefformat{def}{Definition~\ref{#1}}
\newrefformat{cor}{Corollary~\ref{#1}}
\newrefformat{lmm}{Lemma~\ref{#1}}
\newrefformat{prop}{Proposition~\ref{#1}}
\newrefformat{app}{Appendix~\ref{#1}}
\newrefformat{apx}{Appendix~\ref{#1}}
\newrefformat{ex}{Example~\ref{#1}}
\newrefformat{exer}{Exercise~\ref{#1}}
\newrefformat{soln}{Solution~\ref{#1}}

\newcommand{\id}{{\rm id}}

\def\unifto{\mathop{{\mskip 3mu plus 2mu minus 1mu%
	\setbox0=\hbox{$\mathchar"3221$}%
	\raise.6ex\copy0\kern-\wd0%
	\lower0.5ex\hbox{$\mathchar"3221$}}\mskip 3mu plus 2mu minus 1mu}}

\ifx\lesssim\undefined
\def\simleq{{{\mskip 3mu plus 2mu minus 1mu%
	\setbox0=\hbox{$\mathchar"013C$}%
	\raise.2ex\copy0\kern-\wd0%
	\lower0.9ex\hbox{$\mathchar"0218$}}\mskip 3mu plus 2mu minus 1mu}}
\else
\def\simleq{\lesssim}
\fi

\ifx\gtrsim\undefined
\def\simgeq{{{\mskip 3mu plus 2mu minus 1mu%
	\setbox0=\hbox{$\mathchar"013E$}%
	\raise.2ex\copy0\kern-\wd0%
	\lower0.9ex\hbox{$\mathchar"0218$}}\mskip 3mu plus 2mu minus 1mu}}
\else
\def\simgeq{\gtrsim}
\fi

\def\dperp{\indep}

\newtheorem{theorem}{Theorem}
\newtheorem{lemma}[theorem]{Lemma}
\newtheorem{corollary}[theorem]{Corollary}

\newtheorem{proposition}[theorem]{Proposition}

\theoremstyle{definition}
\newtheorem{definition}{Definition}
\newtheorem{example}{Example}
\newtheorem{remark}{Remark}

\newif\ifmapx
{\catcode`/=0 \catcode`\\=12/gdef/mkillslash\#1{#1}}
\edef\jobnametmp{\expandafter\string\csname simple-IMA_apx\endcsname}
\edef\jobnameapx{\expandafter\mkillslash\jobnametmp}
\edef\jobnameexpand{\jobname}
\ifx\jobnameexpand\jobnameapx
\mapxtrue
\else
\mapxfalse
\fi

\long\def\apxonly#1{\ifmapx{\color{blue}#1}\fi}

\begin{document}
\ifpdf
\DeclareGraphicsExtensions{.pgf}
\graphicspath{{figures/}}
\fi

\title{Strong data-processing inequalities for channels and Bayesian networks}

\author{Yury~Polyanskiy\thanks{Y.P. is with the Department of EECS, MIT, Cambridge, MA, \url{yp@mit.edu}. His research has been supported by the Center for Science of Information (CSoI),
an NSF Science and Technology Center, under grant agreement CCF-09-39370 and by the NSF CAREER award under grant
agreement CCF-12-53205.}~~and Yihong Wu\thanks{Y.W. is with
the Department of ECE, University of Illinois at Urbana-Champaign, Urbana, IL, \url{yihongwu@illinois.edu}. His research has been supported in part by NSF grants IIS-1447879, CCF-1423088 and the Strategic Research Initiative of the College of Engineering at the
University of Illinois.}}
\affil{MIT, UIUC}

\date{}
\maketitle

\begin{abstract} 
The data-processing inequality, that is, $I(U;Y) \le I(U;X)$ for a Markov chain $U \to X \to Y$, has been the method of
choice for proving impossibility (converse) results in information theory and many other disciplines. Various
channel-dependent improvements (called strong data-processing inequalities, or SDPIs) of this inequality have been
proposed both classically and more recently.  In this note we first survey known results relating various notions of
contraction for a single channel. Then we consider the basic extension: given SDPI for each constituent channel in a
Bayesian network, how to produce an end-to-end SDPI? 

Our approach is based on the (extract of the) Evans-Schulman
method, which is demonstrated for three different kinds of SDPIs, namely, the usual Ahslwede-G\'acs type contraction
coefficients (mutual information), Dobrushin's contraction coefficients (total variation), and finally the $F_I$-curve
(the best possible non-linear SDPI for a given channel). Resulting bounds on the contraction coefficients are interpreted as
probability of site percolation.
As an example, we demonstrate how to obtain SDPI for an $n$-letter memoryless channel with feedback given an SDPI for $n=1$.

Finally, we discuss a simple observation on the equivalence of a linear SDPI and comparison to an erasure channel (in
the sense of ``less noisy'' order). This leads to a simple proof of a curious inequality of Samorodnitsky (2015), and
sheds light on how information spreads in the subsets of inputs of a memoryless channel.
\end{abstract}

\newpage
\tableofcontents
\newpage


%

\section{Introduction}
Multiplication of a componentwise non-negative vector by a stochastic matrix results in a vector that is
``more uniform''. This observation appears in several classical works~\cite{markov1906extension,doeblin1937cas,GB57} differing in their
particular way of making quantitative estimates. For example, Birkhoff's work~\cite{GB57} initiated a study (sometimes
known as geometric ergodicity) of contraction of the projective distance $d_P(x,y)\eqdef \log \max_{i}{x_i\over y_i} -
\log \min_i {x_i\over y_i}$ between vectors in $\mreals_+^n$. Here, instead, we will be interested in contraction of
statistical distances and information measures involving probability distributions, which we define next.

Fix a transition probability kernel (channel) $P_{Y|X}: \matx\to\maty$ acting between two measurable
spaces.
We denote
by $P_{Y|X}\circ P$ the distribution on $\maty$ induced by the push-forward of the distribution $P$, which is the distribution of the output $Y$ when the input $X$ is distributed according to $P$,
 and by $P \times P_{Y|X}$ the joint distribution $P_{XY}$ if $P_X=P$.
We also denote by $P_{Z|Y} \circ P_{Y|X}$ the serial composition of channels.\footnote{More formally, we should have written $P_{Y|X}:\matp(\matx)\to\matp(\maty)$ as a map between spaces of probability
measures $\matp(\cdot)$ on respective bases. The rationale for our notation $P_{Y|X}:\matx\to\maty$ is that we view Markov kernels as randomized
functions. Then, a single distribution $P$ on $\matx$ is a randomized function acting from a space of a
single point, i.e. $P:[1] \to \matx$, and that in turn explains our notation $P_{Y|X}\circ P$ for denoting the induced
marginal distribution.  }

We define three quantities that will play key role in our discussion: the total variation, the 
Kullback-Leibler (KL) divergence and the mutual information
\begin{align} \TV(P, Q) &\eqdef \sup_E |P[E] - Q[E]| = {1\over 2} \int |\diff P - \diff Q|, \label{eq:tv}\\
	   D(P\|Q) &\eqdef \int \log {\diff P\over \diff Q}\, \diff P, \label{eq:KL}\\
	   I(A;B) &\eqdef D(P_{AB} \| P_A P_{B}).
   \end{align}

The purpose of this paper is to give exposition to the phenomenon that upon passing through a non-degenerate noisy channel
distributions become strictly closer and this leads to a loss of information. Namely we have three effects:
\begin{enumerate}
\item Total-variation (or Dobrushin) contraction: 
 $$ \TV(P_{Y|X} \circ P, P_{Y|X} \circ Q) < \TV(P,Q)\,.$$
\item Divergence contraction:
 $$ D(P_{Y|X} \circ P \| P_{Y|X} \circ Q) < D(P\|Q) $$
\item Information loss: For any Markov chain\footnote{The notation $A\to B\to C$ simply means that $A\dperp C | B$.} $U\to X\to Y$ we

 $$ I(U;Y) < I(U;X)\,. $$
\end{enumerate}
These strict inequalities are collectively referred to as \emph{strong data-processing inequalities} (SDPIs).
The goal of this paper is to show intricate interdependencies between these effects, as well as introducing tools for
quantifying how strict these SDPIs are.

\paragraph{Organization}
In Section~\ref{sec:pp} we overview the case of a single channel. Notably, most of the results in the literature are proved for 
finite alphabets, \ie, $|\calX||\calY|<\infty$, with a few exceptions such as \cite{CKZ98,PW14a}.
We provide in \prettyref{app:general} a self-contained proof of some of these results for general alphabets.

From then on we focus on the question: \textit{Given a multi-terminal network with a single source and multiple sinks, and given
SDPIs for each of the channels comprising the network, how do we obtain an SDPI for the composite channel from source to
sinks?}
It turns out that this question has been addressed implicitly in the work of Evans and Schulman~\cite{evans1999signal}
on redundancy required in circuits of noisy gates. Rudiments also appeared in Dawson~\cite{dawson1975information} as well as Boyen and Koller~\cite{BK98-infer}.

In Section~\ref{sec:mi} we present the essence of the Evans-Schulman method and derive upper bounds on the mutual
information contraction coefficient $\etaKL$ for Bayesian networks (directed graphical models). We also interpret the
resulting bounds as probabilities of disupting end-to-end connectivity under independent removals of graph vertices
(site percolation). Then in
Section~\ref{sec:tv} we derive analogous estimates for Dobrushin's coefficient $\etaTV$ that governs the contraction of
the total variation on networks. While the results exactly parallel those for mutual information, the proof relies on new arguments using coupling.
Finally, Section~\ref{sec:ficurve} extends the technique to bounding the $F_I$-curves (the
non-linear SDPIs). 
Section~\ref{sec:erasure} concludes with an alternative point of view on mutual information contraction, namely that
of comparison to an erasure channel. As an example we give a short proof of a result of
Samorodnitsky~\cite{AS15-noisyfunc} about distribution of information in subsets of channel outputs.

\paragraph{Notation} 
Elements of the Cartesian product $\calX^n$ are denoted $x^n\eqdef (x_1,\ldots,x_n)$ to emphasize
their dimension. Given a transition probability kernel from $P_{Y|X}:\calX\to\calY$ we denote $P_{Y|X}^n=P_{Y^n|X^n}$ the kernel
acting from $\calX^n \to \calY^n$ componentwise independently:
$$ P_{Y^n|X^n}(y^n|x^n) \eqdef \prod_{j=1}^n P_{Y|X}(y_j|x_j). $$
To demonstrate the general bounds we consider the running example of $P_{Y|X}$ being an $n$-letter binary symmetric channel (BSC), given by 
\begin{equation}\label{eq:bsc_def}
	Y = X + Z, \quad X,Y\in \FF_2^n, ~Z\sim \mathrm{Bern}(\delta)^n
\end{equation}
and denoted by $\BSC(\delta)^n$. Throughout this paper $\bar\delta \triangleq 1-\delta$.

\section{SDPI for a single channel}\label{sec:pp}
\subsection{Contraction coefficients for $f$-divergence and mutual information}
Let $f:(0,\infty)\to\mreals$
be a convex function that is strictly convex at 1 and $f(1) = 0$. 
Let $D_f(P||Q) \triangleq \Expect_Q[f(\fracd{P}{Q})]$ denote the $f$-divergence of $P$ and $Q$ with $P\ll Q$,
cf.~\cite{IC67}.\footnote{More generally, $D_f(P||Q) \triangleq \Expect_\mu \qth{f\pth{\frac{\diff P/\diff \mu}{\diff Q/\diff \mu}}
}$, where $\mu$ is a dominating probability measure of $P$ and $Q$, \eg, $\mu=(P+Q)/2$, with the understanding that
$f(0) = f(0+)$, $0f(\frac{0}{0})=0$ and $0 f(\frac{a}{0}) = \lim_{x\downarrow0} x f(\frac{a}{x})$ for $a>0$.} For
example, the total variation \prettyref{eq:tv} and the KL divergence \prettyref{eq:KL} correspond to $f(x)={1\over2}|x-1|$ and $f(x)=x \log x$ respectively; 
taking $f(x)=(x-1)^2$ we obtain the $\chi^2$-divergence: $\chi^2(P\|Q) \eqdef \int ({\diff P\over \diff Q})^2 \diff Q - 1$.

For any $Q$ that is not a point mass, define:
\begin{align}
\eta_f(P_{Y|X},Q) &\eqdef  \sup_{P: 0<D_f(P\|Q)<\infty} {D_f(P_{Y|X} \circ P \| P_{Y|X} \circ Q)\over
			D_f(P\| Q)}, \label{eq:eta_fq}\\
			\eta_f(P_{Y|X}) &\eqdef \sup_Q \eta_f(Q)\,.\label{eq:eta_f} 
   \end{align}
It is easy to show that the supremum is over a non-empty set whenever $Q$ is not a point mass (see \prettyref{app:general}).
   For notational simplicity when the channel is clear from context we abbreviate $\eta_f(P_{Y|X})$ as $\eta_f$.
   For contraction coefficients of total variation, $\chi^2$ and KL divergence, we write
   $\etaTV, \etachi$ and $\etaKL$, respectively, which play prominent roles in this exposition.
   
   One of the main tools for studying ergodicity property of Markov chains as well as Gibbs measures, $\etaTV(P_{Y|X})$ is known as the \emph{Dobrushin's coefficient} of the kernel $P_{Y|X}$.
Dobrushin~\cite{RLD56} showed that the supremum in the definition of $\etaTV$ can be
   restricted to point masses, namely,
   \begin{equation}
	\etaTV(P_{Y|X}) = \sup_{x,x'} \TV(P_{Y|X=x}, P_{Y|X=x'}), 
	\label{eq:etaTV}
\end{equation}
   thus providing a simple criterion for strong ergodicity of Markov processes. 
   Later~\cite[Proposition II.4.10(i)]{CKZ98} (see also \cite[Theorem 4.1]{CIR93} for finite alphabets)
   demonstrated that all other contraction coefficients
   are upper bounded by the Dobrushin's coefficient, with inequality being typically strict (cf.~the BSC example below): 
  \begin{theorem}[{\cite[Proposition II.4.10]{CKZ98}}] For every $f$-divergence, we have
\label{thm:eta_ub}
	\begin{align}\label{eq:eta_ub}
	   	\eta_f(P_{Y|X}) \le \etaTV(P_{Y|X}).
\end{align}	
\end{theorem}   

For the opposite direction, lower bounds on $\eta_f$ typically involves $\etachi$, the contraction coefficient of the $\chi^2$-divergence.
   It is well-known, e.g.~Sarmanov \cite{OS58}, that $\eta_{\chi^2}(P_{Y|X},P_X)$ is the squared second largest
   eigenvalue of the conditional expectation operator, which in turn equals the \emph{maximal correlation}
   coefficient of the joint distribution $P_{XY}$:
   \begin{equation}\label{eq:chi_maxcor0}
	   S(X;Y) \eqdef \sup_{f, g} \rho(f(X), g(Y)) = \sqrt{\eta_{\chi^2}(P_{Y|X},P_X)}\,,
\end{equation}	
where $\rho(\cdot,\cdot)$ denotes the correlation coefficient and 
the supremum is over real-valued functions $f,g$ such that $f(X)$ and $g(Y)$ are square integrable.

The relationship between $\etaKL$ and $\etachi$ on finite alphabets has been systematically studied by Ahlswede and G\'acs~\cite{AG76}. In particular, \cite{AG76} proved
   \begin{equation}\label{eq:chi2KL_bound}
	   \eta_{\chi^2}(P_{Y|X},P_X) \le \etaKL(P_{Y|X},P_X), 
   \end{equation}
   and noticed that the inequality is frequently strict.\footnote{See~\cite[Theorem 9]{AG76} and 
   \cite{AGKN13} for examples.} 
   Furthermore, for finite alphabets, the following equivalence is demonstrated in \cite{AG76}:
 \begin{align}
 \etachi(P_X,P_{Y|X}) < 1 
\iff & ~ \etaKL(P_X,P_{Y|X}) < 1 	\label{eq:C0}\\
\iff & ~ 	\text{graph~} \{(x,y): P_X(x)>0,
   P_{Y|X}(y|x) > 0\}\mbox{~is connected}.
\end{align}
   As a criterion for $\eta_f(P_{Y|X},P_X)<1$, this is an improvement of~\eqref{eq:eta_ub} only for channels with
   $\etaTV(P_{Y|X})=1$. 
The lower bound \prettyref{eq:chi2KL_bound} can in fact be considerably generalized:
\begin{theorem}
\label{thm:eta_lb}
Let $f$ be twice continuously differentiable on $(0,\infty)$ with $f''(1) > 0$. Then for any $P_X$ that is not a point mass,
		\begin{equation}
		\eta_{\chi^2}(P_{Y|X},P_X) \le \eta_f(P_{Y|X},P_X)\,,
		\label{eq:eta_lb-Q}
\end{equation} 
and
		\begin{equation} \label{eq:eta_lb}
		\eta_{\chi^2}(P_{Y|X}) \le \eta_f(P_{Y|X})\,.
\end{equation} 
\end{theorem}
See \prettyref{app:eta_lb} for a proof of \prettyref{eq:eta_lb-Q} for the general case, which yields \prettyref{eq:eta_lb} by taking suprema over $P_X$ on both sides.
Note that \prettyref{eq:eta_lb} (resp.~\prettyref{eq:eta_lb-Q}) have been proved in \cite[Proposition II.6.15]{CKZ98} for the general alphabet (resp.~in \cite[Theorem 3.3]{Raginsky14} for finite alphabets).

Moreover, \prettyref{eq:eta_lb} in fact holds with equality for all
nonlinear and operator convex $f$, \eg, for KL divergence and for squared Hellinger distance; see~\cite[Theorem
1]{CRS94} and \cite[Proposition II.6.13 and Corollary II.6.16]{CKZ98}. Therefore, we have:
\begin{theorem}
\label{thm:chi2KL}
   \begin{equation}\label{eq:chi2KL}
	   	\eta_{\chi^2}(P_{Y|X}) = \etaKL(P_{Y|X})\,.
   \end{equation}   	
\end{theorem}
See \prettyref{app:eta_lb} for a self-contained proof. This result was first obtained in \cite{AG76} using different methods for discrete space.
\apxonly{More generally, the range of 
	$f\mapsto \eta_f$ is included (\textit{and equal} for discrete channels) in
	$$ \sup_{x \neq x'} \beta_1(P_x, P_{x'}) \le \eta_f \le \etaTV $$
	see~\cite[Propositions II.6.3 and II.6.4]{CKZ98}.}%
   	Rather naturally, we also have~\cite[Proposition II.4.12]{CKZ98}:
	$$ \eta_f(P_{Y|X}) = 1 \quad \iff \quad \etaTV(P_{Y|X}) = 1 $$
	for any non-linear $f$.   
   
	   As an illustrating example, for $\BSC(\delta)$ defined in \prettyref{eq:bsc_def},
we have cf.~\cite{AG76} 
\begin{equation}
\eta_{\chi^2}=\etaKL=(1-2\delta)^2<\etaTV=|1-2\delta|.
	\label{eq:etaBSC}
\end{equation}
\prettyref{app:binary} present general results on the contraction coefficients for binary-input arbitrary-output channels, which can be bounded using Hellinger distance within a factor of two.

We next discuss the the fixed-input contraction coefficient $\etaKL(P_{Y|X},Q)$. Unfortunately, there is no simple
reduction to the $\chi^2$-case as in~\eqref{eq:chi2KL}. Besides the lower bound~\eqref{eq:chi2KL_bound}, there is a
variety of upper bounds relating $\etaKL$ and $\eta_{\chi^2}$. We quote~\cite[Theorem 11]{makur2015bounds}, who show for
finite input-alphabet case:
$$ \etaKL(P_{Y|X},Q) \le {1\over \min_x{Q(x)}}\eta_{\chi^2}(P_{Y|X},Q)\,.$$
Another bound (which also holds for all $\eta_f$ with operator-convex $f$) is in~\cite[Theorem 3.6]{Raginsky14}:
$$ \etaKL(P_{Y|X},Q) \le \max\left(\eta_{\chi^2}(P_{Y|X},Q), \sup_{0<\beta<1} \eta_{\LC_\beta}(P_{Y|X},Q)\right)\,,$$
where $\eta_{\LC_\beta}$ denotes contraction coefficient of an $f$-divergence $\LC_\beta(P\|Q) = \beta\bar\beta
\int \frac{(P-Q)^2}{\beta P + \bar\beta Q}$ with $\beta\in(0,1)$ and $\bar\beta=1-\beta$ (see also
Appendix~\ref{app:binary}).

   We also note in passing that SDPIs are intimately related to hypercontractivity and maximal correlation, as
discovered by Ahlswede and G\'acs~\cite{AG76} and recently improved by Anantharam et al.~\cite{AGKN13} and
Nair~\cite{Nair2015}. Indeed, the main result of~\cite{AG76} characterizes $\etaKL(P_{Y|X},P_X)$
   as the maximal ratio of hyper-contractivity of the conditional expectation operator 
   $\EE[\cdot|X]$.

The fixed-input contraction coefficient $\etaKL(Q)$ is closely related to the (modified) log-Sobolev inequalities. Indeed, if $\etaKL(Q)<1$ where $Q$ is the invariant measure for the Markov kernel $P_{Y|X}$, \ie, $P_{Y|X}\circ Q=Q$, then any initial distribution $P$ such that $D(P\|Q)<\infty$
	converges to $Q$ exponentially fast 
	since
	\[
	D(P_{Y|X}^n \circ P || Q) \leq \etaKL^n(P_{Y|X},Q) D(P || Q),
	\]
	where the exponent $\etaKL(P_{Y|X},Q)$ can in turn be
   estimated from log-Sobolev inequalities, e.g.~\cite{MLxx}. When $Q$ is not invariant,
   it was shown~\cite{MLM03} that
   \begin{equation*}
   \label{eq:disc_miclo}
	   	1-\alpha(Q) \le \etaKL(P_{Y|X},Q) \le 1-C\alpha(Q)
   \end{equation*}   
   holds for some universal constant $C$, where $\alpha(Q)$ is a modified log-Sobolev (also known as $1$-log-Sobolev) constant:%
   $$ \alpha(Q) = \inf_{f \neq 1, \|f\|_2=1} {\EE\left[ f^2(X) \log {f^2(X)\over
   f^2(X')}\right] \over \EE[f^2(X) \log f^2(X)] }, \qquad P_{X X'} = Q \times
   (P_{X|Y} \circ P_{Y|X}) .$$
   For further connections between $\etaKL$ and log-Sobolev inequalities on finite alphabets see~\cite{Raginsky13,Raginsky14}.

There exist several other characterizations of $\etaKL$, such as the following one in terms of the contraction of mutual information (cf.~\cite[Exercise III.3.12, p.~350]{CK} for finite alphabet):
\begin{equation}\label{eq:etakl_mi}
		\etaKL(P_{Y|X})  = \sup {I(U; Y)\over I(U;X)}\,,
\end{equation}
where the supremum is over all Markov chains $U\to X \to Y$ with fixed $P_{Y|X}$ (or equivalently, over all joint
distributions $P_{XU}$) such that $I(U;X)<\infty$. This result is an immediate consequence of the following input-dependent version 
(see \prettyref{app:etaKLI} for a proof in the general case; the finite alphabet case has been shown in \cite{AGKN13})
\begin{theorem}
\label{thm:etaKLI}
For any $P_X$ that is not a point mass,
\begin{equation}
	\etaKL(P_{Y|X},P_X) = \sup {I(U; Y)\over I(U;X)} \,,
	\label{eq:etaKL}
\end{equation}
where the supremum is taken over all Markov chains $U\to X \to Y$ with fixed $P_{XY}=P_X\circ P_{Y|X}$ such that $0<I(U;X)<\infty$. 	
\end{theorem}

Another characterization of $\etaKL$, in view of \eqref{eq:chi2KL} and~\eqref{eq:chi_maxcor0}, is 
	$$ \etaKL(P_{Y|X})  = \sup \rho^2(f(X), g(Y))\,, $$
where the supremum is over all $P_X$ 
and real-valued square-integrable $f(X)$ and $g(Y)$.


\subsection{Non-linear SDPI}
How to quantify the information loss if $\etaKL=1$ for the channel of interest? In fact this situation can arise in very basic settings, such as the
additive-noise Gaussian channel under the moment constraint on the input distributions (cf.~\cite[Theorem 9, Section 4.5]{PW14a}), 
where the mutual information does not contract linearly as in \prettyref{eq:etakl_mi}, but can still contract \emph{non-linearly}.
In such cases, establishing a strong-data processing inequality can be done by following the joint-range idea of
Harremo\"es and Vajda~\cite{HV11}. Namely, we aim to find (or bound) the \emph{best possible
data-processing function} $F_I$ defined as follows.

\begin{definition}[$F_I$-curve]
  \label{def:FI}
  Fix $P_{Y|X}$ and define
\begin{equation}
  F_I(t, P_{Y|X})\eqdef \sup_{P_{UX}} \left\{ I(U;Y)\colon I(U;X)\leq t,
    P_{UXY} = P_{UX} P_{Y|X}  \right\}.
\end{equation}
Equivalently, the supremum is taken over all joint distributions  $P_{UXY}$ with a given conditional $P_{Y|X}$ and
satisfying $U\rightarrow X\rightarrow Y$. 
The upper concave envelope of $F_I$ is denoted by $F_I^c$:
$$ F_I^c(t, P_{Y|X}) \eqdef \inf\{f(t):  \forall t'\ge 0 \,\, F_I(t', P_{Y|X}) \le f(t'), f\mbox{--concave}\}\,.$$
Equivalently, we have
\begin{equation}
  F_I^c(t, P_{Y|X})= \sup_{P_{VUX}} \left\{ I(U;Y|V)\colon I(U;X|V)\leq t,
    P_{VUXY} = P_{VUX} P_{Y|X}  \right\}\,,
\end{equation}
where $I(A;B|C) \eqdef I(A,C;B)-I(C;B)$ is the conditional mutual information, and averaging over $V$ serves the role of 
concavification (so that $V$ can be taken binary).
Whenever it does not lead to confusion we will write $F_{Y|X}(t)$ instead of $F_I(t, P_{Y|X})$.
\end{definition}
The operational significance of the $F_I$-curve is that it gives the optimal input-independent strong data processing inequality:
\[
I(U;Y) \leq F_I(I(U;X)),
\]
which generalizes \prettyref{eq:etakl_mi} since $F_I'(0)=\etaKL(P_{Y|X})$ and $t \mapsto \frac{1}{t}F_I(t)$ is decreasing (see, \eg,~\cite[Section I]{CPW15-journal}).
See \cite{CPW15-journal} for bounds and expressions for BSC and Gaussian channels.

Frequently it is more convenient to work with the concavified version $F_I^c$ as it allows for some natural extension of the results about
contraction coefficients. Proposition~\ref{prop:fi_bec} shows that $F_I$ may not be concave.

\subsection{Some applications: classical and new}
The main example of a strong data-processing inequality (SDPI) was discovered by Ahlswede and G\'acs~\cite{AG76}. They have shown, using the characterization \prettyref{eq:C0},
that whenever $P_{Y|X}$ is a discrete memoryless channel that does not admit zero-error communication, we have $\etaKL(P_{Y|X}) \leq \eta <1$ and
\begin{equation}\label{eq:m}
	I(W; Y) \le \eta  I(W; X) 
\end{equation}
for all Markov chains $W\to X\to Y$.

SDPIs have been popular for establishing lower (impossibility) bounds in
various setups, in both classical and more recent works. We mention only a few of these applications:
\begin{itemize}
\item By Dobrushin for showing non-existence of multiple phases in Ising models at high temperatures~\cite{RLD70};
\item By Erkip and Cover in portfolio theory~\cite{EC98};
\item By Evans and Schulman in analysis of noise-resistant circuits~\cite{evans1999signal};
\item By Evans, Kenyon, Peres and Schulman in the analysis of inference on trees and percolation~\cite{EKPS00};
\item By Courtade in distributed data-compression~\cite{TC12};
\item By Duchi, Wainwright and Jordan in statistical limitations of differential privacy~\cite{duchi2013local};
\item By the authors to quantify optimal communication and optimal control in line networks \cite{PW14a};
\item By Liu, Cuff and Verd\'u in key generation~\cite{liu2015secret};
\item By Xu and Raginsky in distributed estimation~\cite{xu2015converses}.
\end{itemize}

All of the applications above use SDPI~\eqref{eq:m} to prove negative (impossibility) statements. A notable exception
is the work of Boyen and Koller~\cite{BK98-infer}, who considered the basic problem of computing the posterior-belief
vector of a hidden Markov model: that is, given a Markov chain $\{X_j\}$ observed over a memoryless channel $P_{Y|X}$,
one aims to recompute $P_{X_j|Y^j_{-\infty}}$ as each new observation $Y_j$ arrives. The problem arises when $X$ is of large
dimension and then for practicality one is constrained to approximate (quantize) the posterior. However, due to
the recursive nature of belief computations, the cumulative effect of these approximations may become overwhelming. 
Boyen and Koller~\cite{BK98-infer} proposed to use the SDPI similar to~\eqref{eq:m} with $\eta<1$ for the Markov chain
$\{X_j\}$  and show that this cumulative
effect stays bounded since $\sum \eta^n < \infty$. Similar considerations also enable one to provide provable
guarantees for simulation of inter-dependent stochastic processes.

\section{Contraction of mutual information in networks}\label{sec:mi}

We start by defining a \emph{Bayesian network} (also known as a \emph{directed graphical model}).
Let $G$ be a finite directed acyclic graph with set of vertices $\{Y_v: v\in\calV\}$ denoting random
variables taking values in a fixed finite alphabet.\footnote{At the expense of technical details, the alphabet can be 
replaced with any countably-generated (e.g.~Polish) measurable space. For clarity of presentation we focus here on finite alphabets.}
We assume that each
vertex $Y_v$ is associated with a conditional distribution $P_{Y_v|Y_{\pa(v)}}$ where $\pa(v)$ denotes parents of $v$, with the
exception of one special ``source'' node $X$ that has no inbound edges (there may be other nodes without inbound edges,
but those have to have their marginals specified). Notice that if $V\subset \calV$ is  an arbitrary set of nodes we can
progressively chain together all the random transformations and unequivocally compute $P_{V|X}$ (here and below we use $V$ and $Y_V=\{Y_v: v\in V\}$ interchangeably). 
We assume that vertices in $\calV$ are topologically sorted so that $v_1 > v_2$ implies there is no path from $v_1$ to
$v_2$. Associated to each node we also define
$$ \eta_v \eqdef \etaKL(P_{Y_v|Y_{\pa(v)}})\,.$$
See the excellent book of Lauritzen~\cite{lau96} for a thorough introduction to a graphical model language of specifying
conditional independencies. 

The following result can be distilled from~\cite{evans1999signal}:
\begin{theorem}\label{th:es} 
Let $W\in \calV$ and $V\subset \calV$ such that $W>V$.
Then
	\begin{equation}
	\etaKL(P_{V,W|X}) \le \eta_W \cdot \etaKL(P_{V, \pa(W) |X}) + (1-\eta_W) \cdot \etaKL(P_{V|X})\,.
	\label{eq:es}
\end{equation}
	Furthermore, let $\perc(V)$ denote the probability that there is a path from $X$ to $V$\footnote{More
	formally, $\perc(V)$ equals probability that there exists a sequence of nodes $v_1,\ldots,v_n$ with $v_1=X$,
	$v_n \in V$ satisfying two conditions: 1) for each $i\in[n-1]$ the pair $(v_i,v_{i+1})$ is a directed edge in
	$G$; and 2) each $v_i$ is not removed.} in the graph if each
	node $v$ is removed independently with probability $1-\eta_v$ (site percolation). Then, we have for every $V\subset \calV$
	\begin{equation}\label{eq:es_perc}
		\etaKL(P_{V|X}) \le \perc(V)\,.
	\end{equation}	
	In particular, if $\eta_v<1$ for all $v\in\calV$ then $\etaKL(P_{V|X})<1$. 
\end{theorem}
\begin{proof} Consider an arbitrary random variable $U$ such that 
	$$ U\to X \to (V, W)\,. $$
	Let $A=\pa(W) \setminus V$. Without loss of generality we may assume $A$ does not contain $X$: indeed, if $A$
	includes $X$ then we can introduce an artificial node $X'$ such that $X'=X$ and include $X'$ into $A$ instead of $X$.  
	Relevant conditional independencies are encoded in the following graph:
	
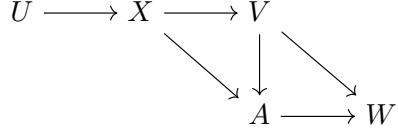
\[
	\begin{tikzcd}
	U \arrow{r} & X \arrow{rd} \arrow{r} 	& V \arrow{d} \arrow{rd} & & \\
								& 					& A \arrow{r}					& W &
	\end{tikzcd}
	\]
	
	From the
	characterization~\eqref{eq:etakl_mi} it is sufficient to show
	\begin{equation}\label{eq:mr1}
			I(U; V, W) \le (1-\eta_W) I(U; V) + \eta_W I(U; V, A)\,. 
	\end{equation}	
	Denote $B = V \backslash \pa(W)$ and $C=V \cap \pa(W)$. Then $\pa(W)=(A,C)$ and $V = (B,C)$. 
	To verify~\eqref{eq:mr1} notice that by assumption we have 
	$$U\to X\to (V,A)\to W \,.$$ 
	Therefore conditioned on $V$ we have the Markov chain 
		$$	U\to X \to A \to W \qquad |V \,$$ 
	and the channel $A\to W$ is a restriction of the original $P_{W|\pa(W)}$ to a subset of the inputs. 
	\apxonly{(To verify, note that $X \to(Y_1,Y_2)\to Z$ implies that $X\to Y_1 \to Z$ conditioned on $Y_2=y_2$ and
	of course $X\to Y_1 \to Z$ do not hold unconditionally in general.)} 
	Indeed, $P_{W|A,V}=P_{W|\pa(W),B}=P_{W|\pa(W)}$ by the assumption of the graphical model. Thus, for every
	realization $v=(b,c)$ of $V$, we have  $P_{W|A=a, V=v} = P_{W|A=a, C=c}$ and therefore
\begin{equation}
I(U; W | V=v) \le \eta(P_{W|A,C=c}) I(U; A | V=v) \leq \eta(P_{W|A,C}) I(U; A | V=v),	
	\label{eq:es-contractv}
\end{equation}
where the last inequality uses the following property of the contraction coefficient which easily follows from either \prettyref{eq:eta_f} or \prettyref{eq:etakl_mi}:
\begin{equation}
\sup_c \eta(P_{W|A,C=c}) \leq \eta(P_{W|A,C}).
	\label{eq:eta-freeze}
\end{equation}
Averaging both sides of \prettyref{eq:es-contractv} over $v\sim P_V$ and using the definition $\eta_W=\eta(P_{W|\pa(W)})=\eta(P_{W|A,C})$, we have
	\begin{equation}
I(U; W | V) \le \eta_W I(U; A | V)\,.
	\label{eq:es-contract}
\end{equation} 
	Adding $I(U; V)$ to both sides yields~\eqref{eq:mr1}.

	We now move to proving the percolation bound~\eqref{eq:es_perc}. First, notice that if a vertex $W$ satisfies
	$W>V$, then letting $\{\exists \, \pi: X \to V\}$ be the event that there exists a directed path from $X$ to (any
	element of) the set $V$ under the site percolation model, we notice that $\{W\text{ removed}\}$ is independent from
	$\{\exists \, \pi: X \to V\}$ and $\{\exists \, \pi: X \to V\cup \pa(W) \}$. Thus we have
	\begin{align*} \perc(V \cup \{W\}) &\eqdef \PP[\exists\, \pi: X \to V \cup \{W\}] \\
				&= \PP[\exists\, \pi: X \to V \cup \{W\}, W\text{ removed}] +  \PP[\exists\, \pi: X \to
				V \cup \{W\}, W\text{ kept}]\\
				&= \PP[\exists\, \pi: X \to V, W\text{ removed}] +  \PP[\exists\, \pi: X \to
				V \cup \pa(W), W\text{ kept}]\\
				&= \PP[\exists\, \pi: X \to V](1-\eta_W) +  \eta_W \PP[\exists\, \pi: X \to
				V \cup \pa(W)]\\
				&= (1-\eta_W) \perc(V) + \eta_W \perc(V \cup \pa(W))\,.
	\end{align*}
	That is, the set-function $\perc(\cdot)$ satisfies the recursion given by the right-hand side of~\eqref{eq:es}.
	Now notice that~\eqref{eq:es_perc} holds trivially for $V=\{X\}$, since both sides are equal to 1. Then, by
	induction on the maximal element of $V$ and applying~\eqref{eq:es} we get that~\eqref{eq:es_perc} holds for all
	$V$.
\end{proof}

Theorem~\ref{th:es} allows us to estimate contraction coefficients in arbitrary (finite) networks by peeling off last nodes
one by one. Next we derive a few corollaries:

\begin{corollary} Consider a fixed (single-letter) channel $P_{Y|X}$ and assume that it is used repeatedly and with perfect feedback 
	to send information from $W$ to $(Y_1,\ldots,Y_n)$. 
	That is, we have for some \textit{encoder} functions $f_j$
	$$ P_{Y^n|W}(y^n|w) = \prod_{j=1}^n P_{Y|X}(y_j|f_j(w, y^{j-1})), $$
	which corresponds to the graphical model:
	$$ \xymatrix{W\ar[r] \ar@/_1pc/[rr] \ar@/_2pc/[rrr] & Y_1 \ar[r] \ar@/^1pc/[rr] & Y_2 \ar[r] &Y_3  \cdots } $$
	Then
	$$ \etaKL(P_{Y^n|W}) \le 1-(1-\etaKL(P_{Y|X}))^n < n \cdot \etaKL(P_{Y|X}) $$
	\label{cor:etaKL-fdbk}
\end{corollary}
\begin{proof}
Apply Theorem~\ref{th:es} $n$ times.	
\end{proof}

\def\scfto{\stackrel{sf}{\to}}
Let us call a path $\pi = (X,\cdots, v)$ with $v\in V$ to be \textit{shortcut-free from $X$ to $V$}, denoted
$X\scfto V$, if there does not
exist another path $\pi'$ from $X$ to any node in $V$ such that $\pi'$ is a subset of $\pi$. (In particular $v$
necessarily is the first node in $V$ that $\pi$ visits.) Also for every path
$\pi=(X, v_1, \ldots, v_m)$ we define
$$ \eta^\pi \eqdef \prod_{j=1}^m \eta_{v_j}\,. $$

\begin{corollary} For any subset $V$ we have
\begin{equation}
\etaKL(P_{V|X}) \le \sum_{\pi: X\scfto V} \eta^\pi\,. 	
	\label{eq:etaKL-es1}
\end{equation}
	In particular, we have the estimate of Evans-Schulman~\cite{evans1999signal}:
	\begin{equation}
	\etaKL(P_{V|X}) \le \sum_{\pi: X\to V} \eta^\pi\,. 
	\label{eq:etaKL-es2}
\end{equation}
	\label{cor:etaKL-es}
\end{corollary}
\begin{proof} Both results are simple consequence of union-bounding the right-hand side of~\eqref{eq:es_perc}. But for
completeness, we give an explicit proof.  First, notice the following two self-evident observations: \begin{enumerate} 
		\item If $A$ and $B$ are disjoint sets of nodes, then
		\begin{equation}\label{eq:mr2}
				\sum_{\pi:X\scfto A\cup B} \eta^\pi = 
			\sum_{\pi:X\scfto A, \text{ avoid~} B } \eta^\pi  + 
			\sum_{\pi:X\scfto B, \text{ avoid~} A} \eta^\pi .
		\end{equation}			
		\item Let $\pi:X\to V$ and $\pi_1$ be $\pi$ without the last node, then
		\begin{align}\label{eq:mr3}
				\pi: X\scfto V \quad \iff\quad \pi_1: X\scfto\{\pa(V)\setminus V\} .
		\end{align}			
	\end{enumerate}
	
	Now represent $V=(V', W)$ with $W>V'$, denote $P=\pa(W)\setminus V$ and assume (by induction) that 
	\begin{align} \etaKL(P_{V'|X}) &\le \sum_{\pi: X\scfto V} \eta^\pi \label{eq:mr4a}\\
		\etaKL(P_{V',P|X}) & \le \sum_{\pi: X\scfto \{V', P\}} \eta^\pi\,.\label{eq:mr4b}
	\end{align}
	By~\eqref{eq:mr2} and~\eqref{eq:mr3} we have 
	\begin{align} \sum_{\pi: X\scfto V}\eta^\pi &= \sum_{\pi: X\scfto V'}\eta^\pi + 
		\sum_{\pi: X\scfto W, \text{ avoid~} V'} \eta^\pi \\
			&= \sum_{\pi: X\scfto V'}\eta^\pi + \eta_W \sum_{\pi: X\scfto P, \text{ avoid~} V'}
			\eta^\pi\label{eq:mr6}
	\end{align}			
	Then by Theorem~\ref{th:es} and induction hypotheses~\eqref{eq:mr4a}-\eqref{eq:mr4b} we get
	\begin{align} \etaKL(P_{V|X}) &\le \eta_W \sum_{\pi: X\scfto \{V', P\}} \eta^\pi + (1-\eta_W)  
		\sum_{\pi: X\scfto V'} \eta^\pi \\
			&= \eta_W \left(\sum_{\pi: X\scfto P, \text{ avoid~} V'} \eta^\pi
		- \sum_{\pi: X\scfto V',\text{ pass~} P} \eta^\pi\right) + \sum_{\pi: X\scfto
		V'}\eta^\pi	\label{eq:mr5}\\
			&\le \eta_W \sum_{\pi: X\scfto P,\text{ avoid~} V'} \eta^\pi + \sum_{\pi: X\scfto
		V'}\eta^\pi \label{eq:mr7}
		\end{align}
		where in~\eqref{eq:mr5} we applied~\eqref{eq:mr2} and split the summation over $\pi:X\scfto V'$ into
		paths that avoid and pass nodes in $P$. Comparing~\eqref{eq:mr6}
		and~\eqref{eq:mr7} the conclusion follows.
\end{proof}

Both estimates \prettyref{eq:etaKL-es1} and \prettyref{eq:etaKL-es2} are compared to that of Theorem~\ref{th:es} in \prettyref{tab:comp} in various graphical models.

\def\vertlabel#1{\vbox to 40pt{\vfil\hbox{#1}\vfil}}
\begin{table}
	\begin{tabular}{c|c|c|c|c}
		Name & Graph & Theorem~\ref{th:es} & \makecell{Estimate \prettyref{eq:etaKL-es1} via \\ shortcut-free paths}  & \makecell{Original Evans-Schulman \\		estimate \prettyref{eq:etaKL-es2}} \\[5pt]
		\hline
		\vertlabel{Markov chain 1} & 
			\vertlabel{$ X\to Y_1 \to B \to Y_2 $} &
			\vertlabel{$\eta$} & \vertlabel{$\eta$} & \vertlabel{$\eta + \eta^3$} \\[5pt]
		\hline
		\vertlabel{Markov chain 2} & 
		\vertlabel{$ \xymatrix{& A\ar[d] \\ X\ar[ur] \ar[r] & B \ar[r] & Y} $ }&
		\vertlabel{$\eta^2$} & \vertlabel{$\eta^2$} & \vertlabel{$\eta^2 + \eta^3$} \\[7pt]
		\hline
		\vertlabel{Parallel channels} & \vertlabel{$ \xymatrix{
					& Y_1 \\
					X \ar[ru]\ar[r] & Y_2 
					}$} 
					& \vertlabel{$2\eta -\eta^2$} & \vertlabel{$2\eta$} & \vertlabel{$2\eta$}
				\\[7pt]
		\hline
		\vertlabel{\parbox{80pt}{Parallel channels with feedback}}  & 
				\vertlabel{$ \xymatrix{
					& Y_1 \ar[d] \\
					X \ar[ru]\ar[r] & Y_2 
					}$} 
					& 
			\vertlabel{$2\eta - \eta^2$} & \vertlabel{$2\eta$} & \vertlabel{$3\eta$}\\[7pt]					
	\end{tabular}
	\caption{Comparing bounds on the contraction coefficient $\etaKL(P_{Y|X})$. For simplicity, we assume that the $\etaKL$ coefficients of all constituent kernels are bounded from above by $\eta$.}\label{tab:comp}
\end{table}


\paragraph{Evaluation for the BSC}
We consider the contraction coefficient for the $n$-letter binary symmetric channel $\BSC(\delta)^n$ defined in \prettyref{eq:bsc_def}. 
By \prettyref{eq:etaBSC}, for $n=1$ we have $\etaKL=(1-2\delta)^2$. Then by Corollary~\ref{cor:etaKL-fdbk} we have
for arbitrary $n$:
\begin{equation}\label{eq:bsc1}
	\etaKL \le  1-(4\delta (1-\delta))^n\,.
\end{equation}

A simple lower bound for $\etaKL$ can be obtained by considering~\eqref{eq:etakl_mi} and taking $U\sim \Bern(1/2)$ and
$U\to X$ being an $n$-letter repetition code, namely, $X=(U,\ldots,U)$. Let\footnote{For elements of $\FF_2^n$, $|\cdot|$ is the Hamming
weight.}
$\epsilon = \PP[|Z| \ge n/2]$ be the probability of error for the maximal likelihood decoding of $U$ based on $Y$, which satisfies the Chernoff bound $\epsilon \leq (4\delta(1-\delta))^{n/2}$.
We have from Jensen's inequality
$$ I(U; Y) = H(U) - H(U|Y) \ge 1-h(\epsilon) = 1 - (4\delta(1-\delta))^{{n\over 2} + O(\log n)}\,,$$
where we used the fact that the binary entropy $h(x)=-x\log x - (1-x) \log (1-x) = -x \log x + O(x^2)$ as $x\to0$.
Consequently, we get
\begin{equation}\label{eq:bsc2}
	\etaKL \ge  1-(4\delta (1-\delta))^{{n\over 2} + O(\log n)}\,.
\end{equation}
Comparing~\eqref{eq:bsc1} and~\eqref{eq:bsc2} we see that $\etaKL \to 1$ exponentially fast. To get the exact 
exponent we need to replace~\eqref{eq:bsc1} by the following improvement:
$$ \etaKL \le \etaTV \le 1-(4\delta (1-\delta))^{{n\over 2} + O(\log n)}\,,$$
where the first inequality is from~\eqref{eq:eta_ub} and the second is from~\eqref{eq:tvbsc3} below. Thus, all in all we
have for $\BSC(\delta)^n$ as $n\to\infty$
\begin{equation}\label{eq:etakl_tight}
	\etaKL,\etaTV=1-(4\delta (1-\delta))^{{n\over 2} + O(\log n)}\,. 
\end{equation}

\section{Dobrushin's coefficients in networks}\label{sec:tv}

The proof of \prettyref{th:es} relies on the characterization \prettyref{eq:etakl_mi} of $\etaKL$ via mutual
information, which satisfies the chain rule. Neither of these two properties is enjoyed by the total variation.
Nevertheless, the following is an exact counterpart of \prettyref{th:es} for total variation. 
\begin{theorem}
Under the same assumption of \prettyref{th:es}, 
\begin{equation}
	\etaTV(P_{V,W|X}) \le (1-\eta_W) \etaTV(P_{V|X}) + \eta_W \etaTV(P_{\pa(W),V|X})\,, 
	\label{eq:es-tvv}
\end{equation}
where $\eta_W = \etaTV(P_{W|\pa(W)})$. 
	Furthermore, let $\perc(V)$ denote the probability that there is a path from $X$ to $V$ in the graph if each
	node $v$ is removed independently with probability $1-\eta_v$ (site percolation). Then, we have for every $V\subset \calV$
	\begin{equation}\label{eq:estv_perc}
		\etaTV(P_{V|X}) \le \perc(V)\,.
	\end{equation}	
	In particular, if $\eta_v<1$ for all $v\in V$, then $\etaTV(P_{V|X}) < 1$.
	\label{thm:es-tv}
\end{theorem}

\begin{proof}
Fix $x, \tx$ and denote by $P$ (resp. $Q$) the distribution conditioned on  $X=x$ (resp. $x'$). Denote  $Z=\pa(W)$.
The goal is to show 
\begin{equation}
	\TV(P_{VW},Q_{VW}) \leq (1-\eta_W) \TV(P_{V},Q_V) + \eta_W \TV(P_{ZV},Q_{ZV})\,.
	\label{eq:goaltv}
\end{equation}
which, by the arbitrariness of $x,x'$ and in view of the characterization of $\eta$ in \prettyref{eq:etaTV}, yields the desired \prettyref{eq:es-tvv}. 
By \prettyref{lmm:tv} in \prettyref{app:tv}, there exists a coupling of $P_{ZV}$ and $Q_{ZV}$, denoted by
$\pi_{ZVZ'V'}$, such that 
\begin{align*} \pi[(Z,V) \neq (Z',V')] &= \TV(P_{ZV},Q_{ZV})\,,\\
	\pi[V \neq V'] &= \TV(P_{V},Q_{V})
\end{align*} 
simultaneously (that is, this coupling is jointly optimal for the total variation of the joint distributions and one pair of
marginals). 

Conditioned on $Z=z$ and $Z'=z'$ and independently of $VV'$, let $WW'$ be distributed according to a maximal coupling of the conditional laws $P_{W|Z=z}$ and $P_{W|Z=z'}$ (recall that $Q_{W|Z} = P_{W|Z} = P_{W|\pa(W)}$ by definition).
 This defines a joint distribution $\pi_{ZVWZ'V'W'}$, under which we have the Markov chain $VV' \to ZZ' \to WW'$.
Then
\begin{equation*}
\pi[W\neq W' | ZVZ'V'] = \pi[W\neq W' | ZZ'] = \TV(P_{W|\pa(W)=Z}, P_{W|\pa(W)=Z'}) \leq \eta_W \indc{Z \neq Z'}.	
	\label{eq:ZZ}
\end{equation*}
Therefore we have
\begin{align*}
\pi[W \neq W'| V=V'] 
= & ~ \Expect[\pi[W \neq W'| ZZ'] | V=V' ]	\nonumber \\
\leq & ~ \eta_W \pi[Z \neq Z' | V=V' ] .
\end{align*}
Multiplying both sides by $\pi[V=V']$ and then adding $\pi[V\neq V']$, we obtain
\begin{align}
\pi[(W,V) \neq (W',V')] 
\leq & ~ 	(1-\eta_W) \pi[V\neq V'] + \eta_W \pi[(Z, V)\neq (Z',V')]\nonumber \\
= & ~ 	(1-\eta_W) \TV(P_V, Q_V) + \eta_W \TV(P_{Z V}, Q_{ZV}) \nonumber, 
\end{align}
where the LHS is lower bounded by $\TV(P_{WV}, Q_{WV})$ and the equality is due to the choice of $\pi$. This yields the
desired \prettyref{eq:goaltv}, completing the proof of~\eqref{eq:es-tvv}. The rest of the proof is done as in
Theorem~\ref{th:es}.
\end{proof}

As a consequence of \prettyref{thm:es-tv}, both \prettyref{cor:etaKL-fdbk} and \ref{cor:etaKL-es} extend to total
variation verbatim with $\etaKL$ replaced by $\etaTV$:
\begin{corollary}\label{cor:tv_par} In the setting of \prettyref{cor:etaKL-fdbk} we have
	\begin{equation}
		\etaTV(P_{Y^n|W}) \le 1-(1-\etaTV(P_{Y|X}))^n < n \cdot \etaKL(P_{Y|X})\,. 
\end{equation}	
\end{corollary}

\begin{corollary} In the setting of \prettyref{cor:etaKL-es} we have
$$ \etaTV(P_{V|X}) \le \sum_{\pi: X\scfto V} \etaTV^\pi \le \sum_{\pi: X\to V} \etaTV^\pi\,, $$
where for any path $\pi = (X, v_1,\ldots, v_m)$ we denoted $\etaTV^\pi \triangleq \prod\limits_{j=1}^m \etaTV(P_{v_j|\pa(v_j)})$.
\end{corollary}

\paragraph{Evaluation for the BSC}
Consider the $n$-letter BSC defined in \prettyref{eq:bsc_def}, where $Y=X+Z$ with $Z\sim\Bern(\delta)^n$ and $|Z|\sim\Binom(n,\delta)$.
By Dobrushin's characterization \prettyref{eq:etaTV}, we have
\begin{align} \etaTV
&= \max_{x,x' \in \FF_2^n} \TV(P_{Y|X=x}, P_{Y|X=x'})  \nonumber \\
&= \TV(\mathrm{Bern}(\delta)^n, \mathrm{Bern}(1-\delta)^n) \nonumber \\
&= \TV(\mathrm{Binom}(n,\delta), \mathrm{Binom}(n,1-\delta)) \label{eq:tvbsc1}\\
		   &= 1-2\PP[|Z|>n/2] - \PP[|Z|=n/2] \label{eq:tvbsc2}\\
		   &= 1-(4\delta (1-\delta))^{{n\over 2} + O(\log n)}\,, \label{eq:tvbsc3}
\end{align}
where \prettyref{eq:tvbsc1} follows from the sufficiency of $|Z|$ for testing the two distributions, \prettyref{eq:tvbsc2} follows from $\TV(P,Q) = 1 - \int P \wedge Q$ and \prettyref{eq:tvbsc3} follows from standard binomial tail estimates (see, \eg, \cite[Lemma 4.7.2]{ash-itbook}). 
The above sharp estimate should be compared to the bound obtained by applying \prettyref{cor:tv_par}:
\begin{equation}\label{eq:bsc3}
	\etaTV \le 1-(2\delta)^n\,.
\end{equation}
Although~\eqref{eq:bsc3} correctly predicts the exponential convergence of $\etaTV\to1$ whenever $\delta < \frac{1}{2}$, the exponent estimated is not
optimal.

\section{Bounding $F_I$-curves in networks}\label{sec:ficurve}

In this section our goal is to produce upper bound bounds on the $F_I$-curve of a Bayesian network $F_{V|X}$ 
in terms of those of the constituent channels. For any vertex $v$ of the network, denote the $F_I$-curve of the channel $P_{v|\pa(v)}$ by $F_{v|\pa(v)}$, abbreviated by $F_{v}$, and the concavified version by $F_v^c$. 

\begin{theorem}\label{th:es-curve} 
In the setting of \prettyref{th:es}, 
\begin{align}
	F_{V,W|X} \le & ~ F_{V|X} + F_W^c \circ (F_{\pa(W),V|X} - F_{V|X})\, ,	\label{eq:es-curve}\\
	F_{V,W|X}^c \le & ~ F_{V|X}^c + F_W^c \circ (F_{\pa(W),V|X}^c - F_{V|X}^c)\, .\label{eq:es-curvec}
\end{align}
Furthermore, the right-hand side of~\eqref{eq:es-curvec} is non-negative, concave,
nondecreasing and upper bounded by the identity mapping $\id$.
\end{theorem}
\begin{remark}
The $F_I$-curve estimate in \prettyref{th:es-curve} implies that of contraction coefficients of \prettyref{th:es}. To see this, note that since $F_{\pa(W),V|X} \leq \id$, the following is a relaxation of \prettyref{eq:es-curve}:
\begin{equation}\label{eq:es-curveX}
\id - F_{V,W|X}  \geq (\id - F_W) \circ (\id - F_{V|X}).
\end{equation}
Consequently, if each channel in the network satisfies an SDPI, then the end-to-end SDPI is also satisfied. That is, if each 
vertex has a non-trivial $F_I$-curve, \ie, $F_{v}< \id$ for all $v\in\calV$,
then the channel $X\to V$ also has a strict contractive property, \ie, $F_{V|X} < \id$. 

Furthermore, since $F_W^c(t) \leq \eta_W t$, noting the fact that $F_{V|X}'(0)=\etaKL(P_{V|X})$ and taking the derivative
on both sides of \prettyref{eq:es-curve} we see that the latter implies \prettyref{eq:es}.
\end{remark}

\begin{proof}
We first show that for any channel $P_{Y|X}$, its $F_{Y|X}$-curve satisfies that $t\mapsto t - F_{Y|X}(t)$ is nondecreasing. Indeed, it is known,
cf.~\cite[Section I]{CPW15-journal}, that $t\mapsto {F_{Y|X}(t)\over t}$ is nonincreasing. Thus, for $t_1 < t_2$ we have
\begin{align*} t_2 - F_{Y|X}(t_2) &\ge t_2 - {t_2\over t_1} F_{Y|X}(t_1) \\
		  &= {t_2\over t_1} \left( t_1 - F_{Y|X}(t_1) \right)\\
		  &\ge t_1 - F_{Y|X}(t_1)\,,
\end{align*}
where the last step follows from the fact that $F_{Y|X}(t)\le t$. Similarly, for any concave function
$\Phi:\mreals_+\to\mreals_+$ s.t. $\Phi(0)=0$ we have ${\Phi(t_2)\over t_2} \le {\Phi(t_1)\over t_1}$. Therefore, the argument above implies $t \mapsto t-\Phi(t)$ is nondecreasing and, in particular, so is $t \mapsto t-F_W^c(t)$.

Let $P_{UX}$ be such that $I(U;X) \leq t$ and $I(U;W,V) = F_{V,W|X}(t)$. 
By the same argument that leads to \prettyref{eq:es-contract} we obtain 
\begin{align*} I(U; W | V=v_0) &\le F_W(I(U; A | V=v_0)) \\
		   &\le F_W^c(I(U; A | V=v_0))\,.  
\end{align*}		   
Averaging over $v_0 \sim P_V$ and applying Jensen's inequality we get
\[
I(U; W , V) \le F_W^c(I(U; \pa(W), V)  - I(U;V)) + I(U;V).
\]
Therefore,
\begin{align}
F_{V,W|X}(t)
\leq & ~ F_W^c(I(U; \pa(W), V)  - I(U;V)) + I(U;V) 	\nonumber \\
\leq & ~ F_W^c(F_{\pa(W), V|X}(t)  - I(U;V)) + I(U;V)		\label{eq:esi1} \\
= & ~ F_{\pa(W), V|X}(t) - (\id - F_W^c)(F_{\pa(W), V|X}(t)  - I(U;V)) 	\nonumber \\
\leq & ~ F_{\pa(W), V|X}(t) - (\id - F_W^c)(F_{\pa(W), V|X}(t)  - F_{V|X}(t)) 		\label{eq:esi2} \\
= & ~ F_{V|X}(t) +  F_W^c(F_{\pa(W), V|X}(t)  - F_{V|X}(t)) \nonumber	\\
\leq & ~ F_{V|X}^c(t) +  F_W^c(F_{\pa(W), V|X}^c(t)  - F_{V|X}^c(t)) \label{eq:esi3}
\end{align}
where 
\prettyref{eq:esi1} and \prettyref{eq:esi2} follow from the facts that $t \mapsto F_W(t)$ and $t \mapsto t - F_W(t)$ are both nondecreasing, and 
\prettyref{eq:esi3} follows from that $a+F_W^c(b-a)$ is nondecreasing in both $a$ and $b$.

Finally, we need to show that the right-hand side of~\eqref{eq:esi3} is nondecreasing and concave (this automatically
implies that~\eqref{eq:esi3} is an upper-bound to the concavification $F_{V|X}^c$). To that end, denote $t_\lambda = \lambda
t_1 + (1-\lambda) t_0$, $f_\lambda = F_{V|X}^c(t_\lambda)$, $g_\lambda=F_{\pa(W), V|X}^c(t_\lambda )$ and notice the
chain
\begin{align}
	f_\lambda +  F_W^c(g_\lambda  - f_\lambda) &\ge 
	\lambda f_1 + (1-\lambda)f_0 +  F_W^c(\lambda (g_1-f_1) + (1-\lambda)(g_0-f_0) ) \label{eq:esi4}\\
	&\ge \lambda (f_1 + F_W^c(g_1-f_1)) + (1-\lambda)(f_0 + F_W^c(g_0-f_0))\label{eq:esi5}
\end{align}	
where~\eqref{eq:esi4} is from concavity of $F_{V|X}^c$, $F_{\pa(W),V|X}^c$ and monotonicity of $(a,b)\mapsto a+F_W^c(b-a)$, and~\eqref{eq:esi5} is from concavity of $F_W^c$.
\end{proof}

\begin{corollary}\label{cor:ficurve} In the setting of \prettyref{cor:etaKL-fdbk} we have
$$ F_{Y^n|W}(t) \le t - \psi^{(n)}(t)\,, $$
where $\psi^{(1)} = \psi$, $\psi^{(k+1)}=\psi^{(k)} \circ \psi$ and  $\psi:\mreals_+\to\mreals_+$ is a
\underline{convex} function such that
$$ F_{Y|X}(t) \le t - \psi(t)\,.$$
\end{corollary}
\begin{proof} 
The case of $n=1$ follows from the assumption on $\psi$. The case of $n>1$ is proved by induction, with the induction step being an
application of Theorem~\ref{th:es-curve} with $V=Y^{n-1}$ and $W=Y_n$.
\end{proof}

Generally, the bound of~\prettyref{cor:ficurve} cannot be improved in the vicinity of zero. As an example where this is
tight, consider a parallel erasure channel, whose $F_I$-curve for $t\le \log q$ is computed in Theorem~\ref{th:fiec}
below.

\paragraph{Evaluation for the BSC}
To ease the notation, all logarithms are with respect to base two in this section.
Let $h(y)=y\log \frac{1}{y}+(1-y)\log \frac{1}{1-y}$ denote the binary entropy function and $h^{-1}:
[0,1] \to [0,{1\over2}]$ its functional inverse. 
Let $p*q \triangleq p(1-q)+q(1-p)$ for $p,q\in[0,1]$ denote binary convolution and define
\begin{equation}\label{eq:psidef}
	\psi(t) \eqdef t-1+h(\delta*h^{-1}(\max(1-t,0))) 
\end{equation}
which is convex and increasing in $t$ on $\reals_+$.
For $n=1$ it was shown in~\cite[Section 2]{CPW15-journal} that the $F_I$-curve of $\BSC(\delta)$ is given
by
$$ F_I(t,\BSC(\delta)) = F_I^c(t,\BSC(\delta))=
t-\psi(t)\,.$$

Applying \prettyref{cor:ficurve} we obtain the following bound on the $F_I$-curve of BSC of blocklength $n$ (even with feedback):
\begin{proposition}\label{prop:samor} Let $Z_1,\ldots,Z_n \iiddistr \Bern(\delta)$ be independent of
$U$. For any (encoder) functions  $f_j, j=1,\ldots,n$, define
	$$ X_j = f_j(U, Y^{j-1}), \quad Y_j = X_j + Z_j\,.$$
Then
\begin{equation}\label{eq:rppr}
		I(U; Y^n) \le I(U; X^n)-\psi^{(n)}(I(U; X^n))\,,
\end{equation}	
where $\psi^{(1)} = \psi$, $\psi^{(k+1)}=\psi^{(k)} \circ \psi$ and  $\psi$ is defined in~\eqref{eq:psidef}. 
\end{proposition}
\begin{remark}
The estimate~\eqref{eq:rppr} was first shown by A. Samorodnitsky (private
communication) under extra technical constraints on the joint distribution of $(X^n,W)$ and in the absence of feedback.
We have then observed that
Evans-Schulman type of technique yields \eqref{eq:rppr} generally.
\end{remark}

Since $\psi(t) = 4\delta(1-\delta) t + o(t) $ as $t\to0$ we get
$$ F_I^c(t,\BSC(\delta)^n) \le t - t (4\delta(1-\delta))^{n + o(n)} $$
as $n\to \infty$ for any fixed $t$. A simple lower bound, for comparison purposes, can be inferred from~\eqref{eq:bsc2}
after noticing that there we have $I(U;X) = 1$, and so
$$ F_I^c(1,\BSC(\delta)^n) \ge 1-(4\delta (1-\delta))^{{n\over 2} + O(\log n)}\,,$$
This shows that the bound of Proposition~\ref{prop:samor} is order-optimal: $F(t)\to t$ exponentially fast. Exact
exponent is given by~\eqref{eq:etakl_tight}.

As another point of comparison, we note the following. 
Existence of capacity-achieving error-correcting codes then easily implies
$$ \lim_{n\to\infty} {1\over n} F_I^c(n\theta,\BSC(\delta)^n) = \min(\theta, C)\,,$$
where $C=1-h(\delta)$ is the Shannon capacity of $\BSC(\delta)$. 
Since for $t>1$ we have $\psi(t)=t-C$ one can show that
$$ \lim_{n\to\infty} {1\over n}\psi^{(n)}(n\theta) = \left|\theta-C\right|^+\,,$$
and therefore we conclude that in this sense the bound~\eqref{eq:rppr} is asymptotically tight.

\section{SDPI via comparison to erasure channels}\label{sec:erasure}

So far our leading example has been the binary symmetric channel~\eqref{eq:bsc_def}. We now consider another important
example:
\begin{example}
For any set $\calX$, the \textit{erasure channel} on $\calX$ with erasure probability $\delta$  is a random transformation from $\calX$ to $\calX\cup\{?\}$, where $?\notin\calX$  defined as
$$ P_{E|X}(e|x) = \begin{cases} \delta, &e=?\\
				1-\delta, &e=x 
	\end{cases}\,. $$
For $\calX=[q]$, we call it the \textit{$q$-ary erasure channel} denoted by $\EC_q(\delta)$. 
In the binary case, we denote the binary erasure channel by $\BEC(\delta)\triangleq \EC_2(\delta)$.
A simple calculation shows that for every $P_{UX}$ we have
\begin{equation}\label{eq:era}
	I(U; E) = (1-\delta) I(U;X) 
\end{equation}
and therefore for $\EC_q(\delta)$ we have $\etaKL(P_{E|X})=1-\delta$ and $F_I(t)=\min((1-\delta)t, \log q)$. 
\end{example}

Next we recall a standard information-theoretic ordering on channels, cf.~\cite[Section 5.6]{ELGK11}:
\begin{definition} Given two channels with common input alphabet, $P_{Y|X}$ and $P_{Y'|X}$, 
we say that $P_{Y'|X}$ is less noisy than $P_{Y|X}$, denoted by $P_{Y|X}\le_{l.n.}P_{Y'|X}$ if for all joint distributions
$P_{UX}$ we have
\begin{equation}\label{eq:ptt_inf}
	I(U;Y) \le I(U; Y')\,.
\end{equation}
\end{definition}

We also have an equivalent formulation in terms of divergence:
\begin{proposition}\label{prop:ptt} $P_{Y|X}\le_{l.n.}P_{Y'|X}$ if and only if for all $P_X,Q_X$ we have
	\begin{equation}\label{eq:ptt_div}
		D(Q_Y\| P_Y) \le D(Q_{Y'} \| P_{Y'}) 
\end{equation}	
	where $P_Y,P_{Y'},Q_Y,Q_{Y'}$ are the output distributions induced by $P_X,Q_X$ over $P_{Y|X}$ and $P_{Y'|X}$, respectively.
\end{proposition}
See Appendix~\ref{app:ptt} for the proof.\footnote{It is tempting to put forward a fixed-$P_X$ version of the previous criterion
(similar to Theorem~\ref{thm:etaKLI}). That would, however, require some extra assumptions on $P_X$. Indeed, knowing that
$I(W;Y)\le I(W;Y')$ for all $P_{W,X}$ with a given fixed $P_X$ tells us nothing about how distributions $P_{Y|X=x}$ and
$P_{Y'|X=x}$ compare outside the support of $P_X$. (For discrete channels and strictly positive $P_X$, however, it is
easy to argue that indeed~\eqref{eq:ptt_div} holds for all $Q_X$ if and only if~\eqref{eq:ptt_inf} holds for all
$P_{U,X}$ with a given marginal $P_X$.)}

The following result shows that the contraction coefficient of KL divergence can be equivalently formulated as being
less noisy than the corresponding erasure channel:\footnote{Note that another popular partial order for random transformations -- that of stochastic degradation -- may also be
related to contraction coefficients, see~\cite[Remark 3.2]{Raginsky14}.}
\begin{proposition}\label{prop:etakl_ec} For an arbitrary channel $P_{Y|X}$ we have
\begin{equation}\label{eq:etakl_ec}
	\etaKL(P_{Y|X}) \le \eta \quad \iff \quad P_{Y|X} \le_{l.n.} P_{E|X}\,,
\end{equation}
where $P_{E|X}$ is the erasure channel on the same input alphabet and erasure probability $1-\eta$.
\end{proposition}
\begin{proof} The definition of $\etaKL(P_{Y|X})$ guarantees for every $P_{UX}$
\begin{equation}\label{eq:era2}
	I(U;Y) \le (1-\delta)I(U;X),
\end{equation}
where the right-hand side is precisely $I(U;E)$ by~\eqref{eq:era}.
\end{proof}

It turns out that the notion of less-noisiness tensorizes:
\begin{proposition}\label{prop:tenso} If $P_{Y_1|X_1} \le_{l.n.} P_{Y'_1|X_1}$ and $P_{Y_2|X_2} \le_{l.n.} P_{Y'_2|X_2}$ then
	$$ P_{Y_1|X_1}\times P_{Y_2|X_2} \le_{l.n.} P_{Y'_1|X_1} \times P_{Y'_2|X_2} $$
	In particular, 
	\begin{equation}\label{eq:uuff0}
		\etaKL(P_{Y|X}) \le \eta \quad \implies \quad P_{Y|X}^n \le_{l.n.} P_{E|X}^n\,. 
\end{equation}	
where $P_{E|X}$ is the erasure channel on the same input alphabet and erasure probability $1-\eta$.
\end{proposition}
\begin{proof}
Construct a relevant joint distribution $U\to X^2 \to (Y^2, Y'^2)$ and consider
\begin{equation}\label{eq:uuff}
	I(U; Y_1,Y_2) = I(U; Y_1) + I(U; Y_2|Y_1)\,.
\end{equation}
Now since $U \dperp Y_2 | Y_1$ we have by $P_{Y_2|X_2} \le_{l.n.} P_{Y'_2|X_2}$
$$ I(U; Y_2|Y_1) \le I(U; Y_2'|Y_1) $$
and putting this back into~\eqref{eq:uuff} we get
$$ I(U; Y_1, Y_2) \le I(U; Y_1) + I(U; Y_2'|Y_1)=I(U; Y_1, Y_2')\,.$$
Repeating the same argument, but conditioning on $Y_2'$ we get
$$ I(U; Y_1, Y_2) \le I(U; Y_1', Y_2')\,,$$
as required. The last claim of the proposition follows from \prettyref{prop:etakl_ec}.
\end{proof}

Consequently, everything that has been said in this paper about $\etaKL(P_{Y|X})$ can be restated in terms of seeking to
compare a given channel in the sense of the $\le_{l.n.}$ order to an erasure channel. It seems natural, then, to
consider erasure channel in somewhat greater details.

\subsection{$F_I$-curve of erasure channels}

\begin{theorem}\label{th:fiec} Consider the $q$-ary erasure channel of blocklength $n$ and erasure probability $\delta$. Its $F_I$-curve is bounded by
\begin{equation}\label{eq:fiec}
	F_I^c(t, \EC_q(\delta)^n) \le \EE[\min(B\log q,t)], \qquad B\sim\Binom(n, 1-\delta) \,.
\end{equation}
The bound is tight in the following cases:
\begin{enumerate}
\item at $t=k\log q$ with integral $k\le n$ if and only if an $(n,k,n-k+1)_q$ MDS code exists\footnote{We remind that a
subset $\matc$ of $[q]^n$ is called an $(n,k,d)_q$ code if $|\matc|=q^k$ and Hamming distance between any two points
from $\matc$ is at least $d$. A code is called maximum-distance separable (MDS) if $d=n-k+1$. This is equivalent to
the property that projection of $\matc$ onto any subset of $k$ coordinates is bijective.}
\item for $t\le \log q$ and $t\ge (n-1)\log q$;
\item for all $t$ when $n=1,2,3$.
\end{enumerate}
\end{theorem}
\begin{remark} Introducing $B'\sim\Binom(n-1,1-\delta)$ and using the identity $\EE[B \indc{B\le a}] =
n(1-\delta)\PP[B'\le a-1]$, we can express the right-hand side of~\eqref{eq:fiec} in terms of binomial CDFs:
$$ \EE[\min(B,x)] = x + \PP[B'\le \lfloor x\rfloor-1] (1-\delta)(n-x) - x\delta \PP[B'\le \lfloor x \rfloor ] $$
This implies that the upper bound \eqref{eq:fiec} is piecewise-linear, increasing and concave.
\end{remark}

\begin{proof} Consider arbitrary $U\to X^n \to E^n$ with $P_{E^n|X^n}=\EC_q(\delta)^n$.
Let $S$ be random subset of $[n]$ which includes each $i\in[n]$ independently with probability $1-\delta$.
A direct computation, shows that 
	\begin{align} I(U; E^n) &= I(U; X_S, S) = \sum_{\sigma \subset [n]} \PP[S=\sigma] I(U; X_\sigma)
	\label{eq:fiec2a} \\
		     &\le \sum_{\sigma \subset [n]} \PP[S=\sigma] \min(|\sigma|\log q, t) = \EE[\min(B\log q,t)]\,.
		     \label{eq:fiec2}
\end{align}
From here~\eqref{eq:fiec} follows by taking supremum over $P_{U,X^n}$. 

Claims about tightness follow by constructing $U=X^n$ and taking $X^n$ to be the output of the MDS code (so that
$H(X_\sigma) = \min(|\sigma| \log q, t)$) and invoking the concavity of $F_I(t)$. One also notes that $[n,1,n]_q$
(repetition code) and $[n, n-1, 2]$ (single parity check code) show tightness at $t=\log q$ and $t=(n-1)\log q$.

Finally, we prove that when $t=k\log q$ and the bound~\eqref{eq:fiec} is tight then a (possibly non-linear)
$(n,k,n-k+1)_q$ MDS code must exist. First, notice that the right-hand side of~\eqref{eq:fiec} is a piecewise-linear and
concave function. Thus the bound being tight for $F_I(t)$ (that is a concave-envelope of $F_I(t)$) should also
be tight as a bound for $F_I(t)$. Consequently, there must exist $U\to X^n\to E^n$ such that the
bound~\eqref{eq:fiec2} is tight with $t=I(U;X^n)$. This implies that we should have 
\begin{equation}\label{eq:fiec3}
	I(U;X_\sigma)= \min(\sigma \log q, t) 
\end{equation}
for all $\sigma \subset [n]$. In particular, we have $I(U;X_i)=\log q$ and thus $H(X_i|U)=0$ and without loss of
generality we may assume that $U=X^n$. Again from~\eqref{eq:fiec3} we have that $H(X^n) = H(X^k)=k \log q$. This implies
that $X^n$ is a uniform distribution on a set of size $q^k$ and projection on any $k$ coordinates is injective. This is
exactly the characterization of an MDS code (possibly non-linear) with parameters $(n, k, n-k+1)_q$.
\end{proof}

We also formulate some interesting observations for binary erasure channels:
\begin{proposition}\label{prop:fi_bec} For $\BEC(n,\delta)$ we have:
\begin{enumerate}
\item For $n\ge 3$ we have that $F_I(t)$ is not concave. More exactly, $F_I(t) < F_I^c(t)$ for $t\in(1,2)$. 
\item For arbitrary $n$ and $t\le \log 2$ or $t\ge (n-1)\log 2$ we have $F_I(t)=F_I^c(t)=\EE[\min(B \log 2,t)]$ with $B$ defined in
in~\eqref{eq:fiec}.
\item For $t=2,n=4$ the bound~\eqref{eq:fiec} is not tight and $F_I^c(t) < \EE[\min(B\log2,t)]$.
\end{enumerate}
\end{proposition}
\begin{proof}

   First note that in Definition~\ref{def:FI} of $F_I(t)$ the supremum is a maximum and 
   and $U$ can be restricted to alphabet of size $|\calX|+2$. So in
   particular, $F_I(t) = f$ if and only if there exists $I(U;Y^n)=f$,
   $I(U;X^n)\le t$.

   Now consider $t\in(1,2)$ and $n=3$ and suppose $(U,X^n)$ achieves the
   bound. For the bound to be tight we must have $I(U;X^3)=t$. For the bound to
   be tight we must have $I(U;X_i)=1$ for all $i$, that is $H(X_i)=1$, $H(X_i|U)=0$ and $H(X^n|U)=0$. Consequently, 
   without loss of generality we may take $U=X^n$. So for the bound to be tight we need to find a distribution s.t.   
   \begin{equation}\label{eq:stran} H(X^3)=H(X_1,X_2)=H(X_2,X_3)=H(X_1,X_3)=t, H(X_1)=H(X_2)=H(X_3)=1 .
\end{equation}   
   It is straightforward to verify that this set of entropies satisfies Shannon inequalities (i.e. submodularity of
   entropy checks), so the main result of~\cite{zhang1997non} shows that there does exist a sequence of triples $X^3$
   (over large alphabets) which attains this point. We will show, however, that this is impossible for binary-valued
   random variables.  First, notice that the set of achievable entropy vectors by binary triplets is a closed subset of
   $\mreals_+^7$ (as a continuous image of a compact set). Thus, it is sufficient to show that~\eqref{eq:stran} itself
   is not achievable.

   Second, note that for any pair $A,B$ of binary random variables with uniform marginals we must have
   $$ A = B+Z, \qquad B\dperp Z \sim \Bern(p)\,.$$
   Without loss of generality, assume that $X_2=X_1+Z$ where $H(Z) = t-1 > 0$. Moreover, $H(X_3|X_1,X_2)=0$ implies that
   $X_3=f(X_1,X_2)$ for some function $f$.
 
   Given $X_1$ we have $H(X_3|X_1=x)=H(X_3|X_2=x)=t-1 > 0$.
   So the function $X_1\mapsto f(X_1,x)$ should not be constant for either
   choice of $x\in\{0,1\}$ and the same holds for $X_2\mapsto f(x,X_2)$. Eliminating cases leaves us
   with $f = X_1 + X_2$ or $f=X_1+X_2+1$. But then $X_3=X_1+X_2=Z$ and  $H(X_3) < 1$, which is a contradiction.
   
   Since by Theorem~\ref{th:fiec} we know that the bound~\eqref{eq:fiec} is tight for $F_I(t)$ we conclude
   that 
   	$$ F_I(t) < F_I^c(t), \quad \forall t\in(1,2)\,. $$

     To show the second claim consider $U=X^n$ and $X_1 = \cdots = X_n \sim \mathrm{Bern}(p)$ for $t\le \log 2$. For $t\ge (n-1)\log2$
     take $X^{n-1}$ to be iid $\mathrm{Bern}({1\over2})$ and 
     $$ X_n = X_1 + \cdots + X_{n-1} + Z\,,$$
     where $Z \sim \mathrm{Bern}(p)$. This yields $I(U;X_\sigma) = H(X_{\sigma})=|\sigma| \log 2$ for every subset
     $\sigma \subset [n]$ of size up to $n-1$. Consequently, the bound~\eqref{eq:fiec} must be tight.

	Finally, third claim follows from Theorem~\ref{th:fiec} and the fact that there is no $[4,2,3]$ binary code,
	e.g.~\cite[Corollary 7, Chapter 11]{macwilliams1977theory}.
\end{proof}

Putting together~\eqref{eq:uuff0} and~\eqref{eq:fiec} we get the following upper bound on the concavified $F_I$-curve of $n$-letter product channels in terms of the contraction coefficient of the single-letter channel. 
\begin{corollary} If $\etaKL(P_{Y|X})=\eta$, then
$$ F_I^c(t, P_{Y|X}^n) \le \EE[\min(B\log q,t)], \qquad B\sim\Binom(n, 1-\delta)\,. $$
\end{corollary}
This gives an alternative proof of~\prettyref{cor:etaKL-fdbk} for the case of no feedback.

\subsection{Samorodnitsky's SDPI}
So far, we have been concerned with bounding the ``output'' mutual information in terms of a certain ``input'' one. However,
frequently, one is interested in bounding some ``output'' information given knowledge of several input ones. For
example, for the parallel channel we have shown that 
$$ I(W; Y^n) \le (1-(1-\etaKL(P_{Y|X}))^n) I(W;X^n)\,.$$
But it turns out that a stronger bound can be given if we have finer  knowledge about the joint distribution of $W$ and $X^n$.

The following bound can be distilled from~\cite{AS15-noisyfunc}:
\begin{theorem}[Samorodnitsky]\label{th:samor} Consider the Bayesian network
$$ U\to X^n \to Y^n\,,$$
where $P_{Y^n|X^n} = \prod_{i=1}^n P_{Y_i|X_i}$ is a memoryless channel with $\eta_i \eqdef \eta_{KL}(P_{Y_i|X_i})$.
Then we have
\begin{equation} I(U; Y^n) \le I(U; X_S|S) = I(U; X_S,S)\,,\label{eq:samor}
\end{equation}
where $S\dperp(U,X^n,Y^n)$ is a random subset of $[n]$ generated by independently sampling each element $i$ with
probability $\eta_i$. In particular, if $\eta_i=\eta$ for all $i$, then
\begin{equation}\label{eq:samor2}
	I(U; Y^n) \le \sum_{\sigma \subset [n]} \eta^{|\sigma|} (1-\eta)^{n-|\sigma|} I(U; X_\sigma) 
\end{equation}
\end{theorem}
\begin{proof} Just put together characterization~\eqref{eq:etakl_ec}, tensorization property~\prettyref{prop:tenso} to
get $I(U;Y^n) \le I(U; E^n)$, where $E^n$ is the output of the product of erasure channels with erasure probabilities
$1-\eta_i$. Then the calculation~\eqref{eq:fiec2a} completes the proof.
\end{proof}
\begin{remark} Let us say that ``total'' information $I(U;X^n)$ is distributed among subsets
of $[n]$ as given by the following numbers: 
$$ I_k \eqdef {n \choose k}^{-1} \sum_{T\in {[n] \choose k}} I(U; X_T)\,.$$
Then bound~\eqref{eq:samor2} says (replacing $\Binom(n,\eta)$ by its mean value $\eta n$):
$$ I(U;Y^n) \simleq I_{\eta n}\,. $$
Informally: the only kind of information about $U$ that has a chance to be inferred on the basis of $Y^n$ is one that is
contained in subsets of $X$ of size at most $\eta n$. 
\end{remark}
\begin{remark} Another implication of the Theorem is a strengthening of the Mrs.~Gerber's Lemma. Fix a
single-letter channel $P_{Y|X}$ and suppose that for some increasing \textit{convex} function $m(\cdot)$ and all random variables
$X$ we have
$$ H(Y) \ge m(H(X))\,.$$
Then, in the setting of the Theorem we have
\begin{equation}\label{eq:mgl_new}
	H(Y^n) \ge m\left( {1\over \eta n} H(X_S |S)\right)\,.
\end{equation}
Note that by Han's inequality~\eqref{eq:mgl_new} is strictly better than the simple consequence of the chain rule:
$H(Y^n) \ge n m(H(X^n)/n)$. For the case of $P_{Y|X}=\BSC(\delta)$ the bound~\eqref{eq:mgl_new} is a sharpening of the
Mrs.~Gerber's Lemma, and has been the focus of~\cite{AS15-noisyfunc}, see
also~\cite{OO16-isit}. To prove~\eqref{eq:mgl_new} let $X^n \to E^n$ be $\EC(1-\eta)$. Then, by \prettyref{th:samor} applied to
$U=X_i$, $n=i-1$ we have
$$ H(X_i|Y^{i-1}) \ge  H(X_i| E^{i-1})\,.$$
Thus, from the chain rule and convexity of $m(\cdot)$ we obtain
$$ H(Y^n) = \sum_i H(Y_i|Y^{i-1}) \ge n m\left({1\over n} \sum_i H(X_i |E^{i-1})\right)\,,$$
and the proof is completed by computing $H(E^n)$ in two ways:
\begin{align*} nh(\eta) + H(X_S|S) &= H(E^n) \\
	&= \sum_i H(E_i |E^{i-1}) = \sum_i h(\eta) + \eta H(X_i|E^{i-1})\,.
\end{align*}	
\end{remark}

\begin{remark} Using Proposition~\ref{prop:ptt} we may also state a divergence version of the Theorem: In the setting of
Theorem~\ref{th:samor} for any pair of distributions $P_{X^n}$ and $Q_{X^n}$ we have
	$$ D(P_{Y^n} \| Q_{Y^n}) \le D(P_{X_S|S} \| Q_{X_S|S}|P_S)\,.$$
Similarly, we may extend the argument in the previous remark: If for a fixed $Q_X,Q_Y$ (not necessarily related by
$P_{Y|X}$) there exists an increasing concave function
$f$ such that for all $P_X$ and $P_Y=P_{Y|X}\circ P_X$ we have
	$$ D(P_X \| Q_X) \le f(D(P_Y \| Q_Y)) \qquad \forall P_X$$
	then
	$$ D(P_{Y^n} \| (Q_Y)^n) \le n f\left({1\over \eta n} D(P_{X_S|S} \| \prod_{i\in S} Q_{X}|P_S)\right)\,.$$
\end{remark}
\apxonly{
\begin{proof}[old] Case of $n=1$ follows from~\eqref{eq:etaKL}. Then proceed by induction. A single step of
Evans-Schulman~\eqref{eq:mr1} yields
\begin{align} I(U;Y^n) &= I(U; Y^{n-1}) + I(U;Y_n | Y^{n-1}) \\
	&\le (1-\eta_n) I(U; Y^{n-1}) + \eta_n I(U; X_n , Y^{n-1})\,.\label{eq:mr1X}
\end{align}	
Now, letting $S'$ be a random subset of $[n-1]$ generated as $\prod_{i=1}^{n-1}\mathrm{Bern}(\eta_i)$ and applying induction
hypothesis for the second term we have\apxonly{\footnote{In first line we use the fact that there is no feedback.}}
\begin{align} I(U; X_n, Y^{n-1}) &= I(U; X_n) + I(U; Y^{n-1}|X_n)\\
	&\le I(U; X_n) + I(U; X_{S'}|S', X_n) \\
	&= I(U; X_n) + I(U; X_{S' \cup \{n\}} | S') - P[S' \neq \emptyset] I(U; X_n)\\
	&= \PP[S' = \emptyset] I(U; X_n) + I(U; X_{S' \cup \{n\}}|S')\,.\label{eq:mr2X}
\end{align}
Plugging~\eqref{eq:mr2X} into~\eqref{eq:mr1X} and applying induction hypothesis to $I(U;Y^{n-1})$ we get
\begin{align} 
	I(U;Y^n) &\le (1-\eta_n)I(U; X_{S'}|S') + \eta_n \PP[S' = \emptyset] I(U; X_n) + \eta_n I(U; X_{S' \cup \{n\}}|S')
\end{align}
It is clear that right-hand side of the latter equals~\eqref{eq:samor}.
\end{proof}
}

\section*{Acknowledgment}

We thank Prof.~M.~Raginsky for references~\cite{BK98-infer,dawson1975information,SG79} and Prof.~A.~Samorodnitsky for
discussions on Proposition~\ref{prop:samor} with us. We also thank Aolin Xu for pointing out~\eqref{eq:etakl_tight}.
We are grateful to an anonymous referee for helpful comments.

\appendix

\section{Contraction coefficients on general spaces}
	\label{app:general}

\subsection{Proof of \prettyref{thm:eta_lb}}
	\label{app:eta_lb}
	We show that 
	\begin{equation}
	\eta_f(P_{Y|X},P_X) = \sup_{Q_X} \frac{D_f(Q_{Y} \| P_{Y})}{D_f(Q_X\|P_X)} \geq \etachi(P_{Y|X},P_X) = \sup_{Q_X} \frac{\chi^2(Q_{Y} \| P_{Y})}{\chi^2(Q_X\|P_X)},
	\label{eq:KLchi-P}
\end{equation}
where both suprema are over all $Q_X$ such that the respective denominator is in $(0,\infty)$. With the assumption that
$P_X$ is not a point mass, namely, there exists a measurable set $E$ such that $P_X(E) \in (0,1)$, it is clear that such
$Q_X$ always exists. For example, let $Q_X = \frac{1}{2}(P_X+P_{X|X\in E})$, where $P_{X|X\in E}(\cdot) \eqdef
{P_X(\cdot \cap E)\over P_X(E)}$. Then $\frac{1}{2} \leq \fracd{Q_X}{P_X} \leq \frac{1}{2}(1+\frac{1}{P_X(E)})$ and hence $D_f(Q_X\|P_X) < \infty$ since $f$ is continuous. Furthermore, $Q_X \neq P_X$ implies that $D_f(Q_X\|P_X) \neq 0$ \cite{IC67}. 

	The proof follows that of \cite[Theorem 5.4]{CIR93} using the local quadratic behavior of $f$-divergence; however, in order to deal with general alphabets, additional approximation steps are needed to ensure the likelihood ratio is bounded away from zero and infinity. 
	
	Fix $Q_X$ such that $\chi^2(Q_X\|P_X) < \infty$. Let $A=\{x: \fracd{Q_X}{P_X}(x) < a\}$ where $a>0$ is sufficiently large such that $Q_X(A) \geq 1/2$. 
Let $ Q_X' = Q_{X|X\in A}$ and $Q'_Y = P_{Y|X} \circ Q_X'$. Then $\fracd{Q'_Y}{P_Y} \leq \frac{a}{Q_X(A)} \leq 2a$.
Let $Q_X'' = \frac{1}{a} P_X + (1-\frac{1}{a}) Q_X'$ and $Q_Y''= P_{Y|X} \circ Q_X' = \frac{1}{a} P_Y + (1-\frac{1}{a}) Q_Y'$.
Then we have 
\begin{equation}
	\frac{1}{a} \leq \fracd{Q_X''}{P_X} \leq 2a + \frac{1}{a}, \quad \frac{1}{a} \leq \fracd{Q_Y''}{P_Y} \leq 2a + \frac{1}{a}.
	\label{eq:bddLLR}
\end{equation}
Note that
$\chi^2(Q_X'\|P_X) = \frac{1}{Q(X\in A)} \Expect_P[ (\fracd{Q_X}{P_X})^2 \indc{X\in A}]-1$. By dominated convergence theorem, 
$\chi^2(Q_X'\|P_X) \to \chi^2(Q_X\|P_X)$ as $a\diverge$.
On the other hand,
	since $Q'_Y \to Q_Y$ pointwise, the weak lower-semicontinuity of $\chi^2$-divergence yields
	 $\liminf_{a\to\infty} \chi^2(Q_Y'\|P_Y) \ge \chi^2(Q_Y\|P_Y)$. 
Furthermore, using the simple fact that $\chi^2(\epsilon P + (1-\epsilon) Q \|P) = (1-\epsilon)^2 \chi^2(Q\|P)$, we have $\frac{\chi^2(Q_X''\|P_X)}{\chi^2(Q_Y''\|P_Y)} = \frac{\chi^2(Q_X'\|P_X)}{\chi^2(Q_Y'\|P_Y)}$.
Therefore, to prove \prettyref{eq:KLchi-P}, it suffices to show for each fixed $a$, for any $\delta>0$, there exists $\tilde P_X$ such that 
$\frac{D_f(\tilde P_{X} \| \tilde P_{Y})}{D_f(Q_X\|P_X)} \geq \frac{\chi^2(Q_X''\|P_X)}{\chi^2(Q_Y''\|P_Y)} - \delta$.

For $0 < \epsilon < 1$, let $\tilde P_X = \bar{\epsilon} P_X + \epsilon Q_X''$, which induces $\tilde P_Y = P_{Y|X} \circ \tilde P_X = \bar{\epsilon} P_Y + \epsilon Q_Y''$. Then
$D_f(\tilde P_X\|P_X) = \Expect_{P_X}[f(1 + \epsilon (\fracd{Q_X''}{P_X}-1))]$. Recall from \prettyref{eq:bddLLR} that
$\fracd{Q_X''}{P_X} \in [\frac{1}{a},\frac{1}{a}+2a]$. Since $f''$ is continuous and $f''(1)=1$, by Taylor's theorem and dominated convergence theorem, we have
$D_f(\tilde P_X\|P_X) = \frac{\epsilon^2 }{2} \chi^2(Q_X''\|P_X) (1+o(1))$. Analogously, $D_f(\tilde P_Y\|P_Y) = \frac{\epsilon^2 }{2} \chi^2(Q_Y''\|P_Y) (1+o(1))$.
This completes the proof of $\eta_f(P_X) \geq \etachi(P_X)$.

\begin{remark}
In the special case of KL divergence, we can circumvent the step of approximating by bounded likelihood ratio: 
By \cite[Lemma 4.2]{PWLN}, since $\chi^2(Q_Y\|P_Y) \leq \chi^2(Q_X\|P_X) < \infty$, we have
$D(\tilde P_X\|P_X) = \epsilon^2 \chi^2(Q_X\|P_X)/2 + o(\epsilon^2)$ and $D(\tilde P_Y\|P_Y) = \epsilon^2 \chi^2(Q_Y\|P_Y)/2 + o(\epsilon^2)$, as $\epsilon \to 0$. Therefore $\frac{\chi^2(Q_{Y} \| P_{Y})}{\chi^2(Q_X\|P_X)} \leq \lim_{\epsilon\to0} \frac{D(\tilde P_Y\|P_Y)}{D(\tilde P_X\|P_X)} \leq \etaKL(P_X)$. Therefore $\etaKL(P_X) \geq \etachi(P_X)$
\end{remark}

\subsection{Proof of \prettyref{thm:chi2KL}}
	\label{app:chi2KL}
We prove
\begin{equation}
		\etaKL = \etachi.
	\label{eq:KLchi}
\end{equation}
First of all, $\etaKL \geq \etachi$ follows from \prettyref{thm:eta_lb}.
For the other direction we closely follow the argument of \cite[Theorem 1]{CRS94}. Below we prove the following integral representation:
\begin{equation}
	D(Q\|P) = \int_0^\infty \chi^2(Q\|P^t) \diff t,
	\label{eq:KL-integral}
\end{equation}
where $P^{t} \triangleq \frac{tQ + P}{1+t}$. Then
\begin{align*}
D(Q_Y\|P_Y) 
= & ~ \int_0^\infty \chi^2(Q_Y\|P_Y^t)  \diff t	\\
\leq & ~ \int_0^\infty \etachi \cdot \chi^2(Q_X\|P_X^t)  \diff t = \etachi D(Q_X\|P_X). 
\end{align*}
where we used $P_Y^t=P_{Y|X} \circ P_X^t$.
It remains to check \prettyref{eq:KL-integral}. Note that 
\[
- \log x =  
\int_0^\infty \frac{1-x}{(x+t)(1+t)} \diff t
\]
Therefore
\[
D(Q\|P) = \int_0^\infty \frac{1}{1+t} \expects{\frac{\diff Q - \diff P}{\diff P +t \diff Q}}{Q} \diff t
\]
Note that $t \expects{\frac{\diff Q - \diff P}{\diff P +t \diff Q}}{Q} = - \expects{\frac{\diff Q - \diff P}{\diff P +t \diff Q}}{P}$. Therefore
$\expects{\frac{\diff Q - \diff P}{\diff P +t \diff Q}}{Q} = \frac{1}{1+t} \int \frac{(\diff Q - \diff P)^2}{\diff P + t \diff Q} = (1+t) \chi^2(Q\|P^t)$, completing the proof of \prettyref{eq:KL-integral}.

It is instructive to remark how this result was established for finite alphabets originally in~\cite{AG76}. Consider the
map
$$ P_X \mapsto V_r(P_X, Q_X) \eqdef D(P_{Y|X}\circ P_X \| P_{Y|X} \circ Q_X) - r D(P_X\|Q_X)\,. $$
A simple differentiation shows that Hessian of this map at $P_X$ is negative-definite if and only if $r>\eta_{\chi^2}(P_{Y|X}, P_X)$ and
negative semidefinite if and only if $r\ge \eta_{\chi^2}(P_{Y|X}, P_X)$ (note that this does not depend on $Q_X$). Thus, taking
$r=\eta_{\chi^2}(P_{Y|X})$ the map $P_X \mapsto V_r(P_X,Q_X)$ is concave in $P_X$ for all $Q_X$. Thus, its local extremum at $P_X=Q_X$
is a global maximum and hence $V_r(P_X,Q_X)\le 0$.

\subsection{Proof of \prettyref{thm:etaKLI}}
	\label{app:etaKLI}
	We shall assume that $P_X$ is not a point mass, namely, there exists a measurable set $E$ such that $P_X(E) \in (0,1)$. 
	Define
\[
\etaKL(P_X) = \sup_{Q_X} \frac{D(Q_{Y} \| P_{Y})}{D(Q_X\|P_X)}
\]
where the supremum is over all $Q_X$ such that $0<D(Q_X\|P_X) < \infty$.
It is clear that such $Q_X$ always exists (\eg, $Q_X = P_{X|X\in E}$ and $D(Q_X\|P_X) = \log \frac{1}{P_X(E)} \in (0,\infty)$). Let
\[
\eta_I(P_X) = 	\sup \frac{I(U; Y)}{I(U;X)}
\]
where the supremum is over all Markov chains $U\to X \to Y$ with fixed $P_{XY}$ such that $0<I(U;X)<\infty$. Such Markov chains always exist, \eg, $U=\indc{X\in E}$ and then $I(U;X)=h(P_X(E)) \in (0,\log 2)$.
The goal of this appendix is to prove~\eqref{eq:etaKL}, namely
$$ \etaKL(P_X) = \eta_I(P_X)\,.$$

The inequality $\eta_I(P_X) \leq \etaKL(P_X)$ follows trivially:
$$ I(U;Y) = D(P_{Y|U} \| P_Y |P_U) \le \etaKL(P_X) D(P_{X|U}\|P_X | P_U) = \etaKL(P_X) I(X;U)\,.$$
For the other direction, fix $Q_X$ such that $0<D(Q_X\|P_X) < \infty$. First, consider the case where $\fracd{Q_X}{P_X}$
is bounded, namely, $\fracd{Q_X}{P_X} \leq a$  for some $a>0$ $Q_X$-a.s. 
For any $\epsilon \leq \frac{1}{2a}$, let $U \sim \Bern(\epsilon)$ and define the probability measure $\tilde P_X = \frac{P_X-\epsilon Q_X}{1-\epsilon}$. Let $P_{X|U=0} = \tilde P_X$ and $P_{X|U=1} = Q_X$, which defines a Markov chain $U \to X \to Y$ such that $X,Y$ is distributed as the desired $P_{XY}$. Note that 
\[
\frac{I(U; Y)}{I(U;X)} = \frac{\bar \epsilon D(\tilde P_Y\|P_Y) + \epsilon D(Q_Y\|P_Y)}{\bar \epsilon D(\tilde P_X\|P_X) + \epsilon D(Q_X\|P_X)}
\]
where $\bar\epsilon=1-\epsilon$ and $\tilde P_Y  = P_{Y|X} \circ \tilde P_X$.
We claim that 
\begin{equation}
	D(\tilde P_X\|P_X) = o(\epsilon),
	\label{eq:oepsilon}
\end{equation} which, in view of the data processing inequality $D(\tilde P_X\|P_X) \leq D(\tilde P_Y\|P_Y)$, implies $\frac{I(U; Y)}{I(U;X)} \xrightarrow{\epsilon \downarrow 0} \frac{D(Q_Y\|P_Y)}{D(Q_X\|P_X)}$ as desired. 
To establish \prettyref{eq:oepsilon}, define the function
$$ f(x,\epsilon) \eqdef \begin{cases} {1-\epsilon x\over \epsilon(1-\epsilon)} \log{1-\epsilon x\over 1-\epsilon}\,, &
\epsilon > 0 \\
	(x-1) \log e, &\epsilon=0\,. \end{cases} $$
One easily notices that $f$ is continuous on $[0,a]\times[0,{1\over 2a}]$ and thus bounded. So we get, by
bounded convergence theorem,
$$ \frac{1}{\epsilon} D(\tilde P_X\|P_X) = \EE_{P_X}\left[f\left(\fracd{Q_X}{P_X}, \epsilon\right) \right] \to  \EE_{P_X}\left[\fracd{Q_X}{P_X}-1\right] \log e = 0\,. $$

To drop the boundedness assumption on $\fracd{Q_X}{P_X}$ we simply consider the conditional distribution $ Q_X' \eqdef Q_{X|X\in A}$
where $A=\{x: \fracd{Q_X}{P_X}(x) < a\}$ and $a>0$ is sufficiently large so that $Q_X(A) > 0$. Clearly, as $a\diverge$, we have $Q_X' \to Q_X$ and $Q'_Y \to Q_Y$
pointwise (i.e. $Q'_Y(E)\to Q_Y(E)$ for every measurable set $E$), where $Q'_Y \triangleq P_{Y|X} \circ Q'_X$. Hence the lower-semicontinuity of divergence yields
$$ \liminf_{a\to\infty} D(Q_Y'\|P_Y) \ge D(Q_Y\|P_Y)\,.$$
Furthermore, since $\fracd{Q'_X}{P_X} = \frac{1}{Q_X(A)} \fracd{Q_X}{P_X} \Indc_{A}$, we have
\begin{align}
D(Q'_X\|P_X)
= & ~ \log \frac{1}{Q_X(A)} + \frac{1}{Q_X(A)} \EE_Q\left[\log \fracd{Q_X}{P_X} \mathbf{1}\left\{\fracd{Q_X}{P_X} \leq a\right\}\right].
\end{align}
Since $Q_X(A)\to 1$, by dominated convergence (note: $\EE_Q[|\log \fracd{Q_X}{P_X}|] < \infty$) we have
$ D(Q'_X\|P_X) \to D(Q_X\|P_X) $.
Therefore,
$$ \liminf_{a\to\infty} {D(Q_Y'\|P_Y)\over D(Q_X'\|P_X)} \ge {D(Q_Y\|P_Y)\over D(Q_X\|P_X)}\,,$$
completing the proof.

\subsection{Proof of Proposition~\ref{prop:ptt}}\label{app:ptt}
First, notice the following simple result: 
\begin{align}\label{eq:kll2} 
	D(Q \|\lambda P + \bar\lambda Q) = o(\lambda), \lambda\to0 \quad \iff \quad P \ll Q 
\end{align}	
Indeed, if $P \not \ll Q$ then there is a set $E$ with $p=P[E]>0=Q[E]$. Denote the binary divergence by 
$d(p\|q) \triangleq D(\Bern(p)\|\Bern(q))$. 
Applying data-processing for
divergence to $X \mapsto 1_E(X)$, we get
$$ D(Q\| \lambda P + \bar\lambda Q) \ge d(0 \| \lambda p) = \log{1\over 1-\lambda p} $$
and the derivative at $\lambda\to0$ is non-zero. If $P\ll Q$, then let $f = {dP\over dQ}$ and notice
	$$ \log \bar\lambda \le \log (\bar \lambda + \lambda f) \le \lambda(f-1)\log e\,.$$
	Dividing by $\lambda$ and assuming $\lambda < {1\over 2}$ we get
		$$ \left| {1\over \lambda} \log(\bar \lambda + \lambda f) \right| \le C_1 f + C_2\,,$$
		for some absolute constants $C_1,C_2$.
	Thus, by the dominated convergence theorem we get
		$$ {1\over \lambda} D(Q \|\lambda P + \bar\lambda Q)  = - \int dQ \left({1\over \lambda} \log(\bar
		\lambda + \lambda f) \right) \to \int dQ(1-f) = 0\,.$$
Another observation is that  
	\begin{equation}\label{eq:kll3}
		\lim_{\lambda \to 0} D(P \|\lambda P + \bar\lambda Q) = D(P\|Q)\,,
\end{equation}	
	regardless of the finiteness of the right-hand side (this is a property of all convex lower-semicontinuous
	functions).

Now, we prove Proposition~\ref{prop:ptt}. One direction is easy: if $D(Q_Y\|P_Y)\le D(Q_{Y'} \| P_{Y'})$ then
	$$ I(W;Y) = D(P_{Y|W} \|P_Y |P_W) \le D(P_{Y'|W} \| P_{Y'} | P_W) = I(W;Y')\,.$$
For the other direction, consider an arbitrary pair $(P_X,Q_X)$. Let $W=\Bern(\epsilon)$ and $P_{X|W=0}=P_X$,
$P_{X|W=1}=Q_X$. Then, we get
	$$ I(W;Y) = \bar\epsilon D(P_Y\|\bar \epsilon P_Y + \epsilon Q_Y) + \epsilon D(Q_Y\| \bar \epsilon P_Y +
	\epsilon Q_Y)\,,$$
	and similarly for $I(W;Y')$. Assume that $D(Q_{Y'} \| P_{Y'})<\infty$, for otherwise~\eqref{eq:ptt_div} holds trivially. 
	Then $Q_{Y'} \ll P_{Y'}$ and we get from~\eqref{eq:kll2} and~\eqref{eq:kll3} that
	\begin{equation}\label{eq:ptt3}
		I(W;Y') = \epsilon D(Q_{Y'}\|P_{Y'}) + o(\epsilon)\,.
\end{equation}	
	On the other hand, again from~\eqref{eq:kll3}
	\begin{equation}\label{eq:ptt4}
		I(W;Y) \ge \epsilon D(Q_Y\| \bar \epsilon P_Y + \epsilon Q_Y) = \epsilon D(Q_Y\|P_Y) + o(\epsilon)\,. 
\end{equation}	
	Since by assumption $I(W;Y) \le I(W;Y')$ we conclude from comparing~\eqref{eq:ptt3} to~\eqref{eq:ptt4} that
	$D(Q_Y\|P_Y) \le D(Q_{Y'}\| P_{Y'}) < \infty$, completing the
	proof.

\section{Contraction coefficients for binary-input channels}
\label{app:binary}
In this appendix we provide a tight characterization of the KL contraction coefficient for binary-input channel $P_{Y|X}$, where $X\in\{0,1\}$ and $Y$ is arbitrary. 
Clearly, $\etaKL(P_{Y|X})$ is a function of $P \triangleq P_{Y|X=0}$ and $Q \triangleq P_{Y|X=1}$, which we abbreviate
as $\eta(\{P,Q\})$.
The behavior of this quantity closely resembles that of divergence between distributions. Indeed, we expect
$\eta(\{P,Q\})$ to be bigger if $P$ and $Q$ are more dissimilar and, furthermore, 
$\eta(\{P,Q\})=0$ (resp.~$1$)  if and only if $P=Q$ (resp.~$P \perp Q$).
Next we show that $\eta(\{P,Q\})$ is essentially equivalent to Hellinger distance:
\begin{theorem} Consider a binary input channel $P_{Y|X}:\{0,1\}\to \maty$ with $P_{Y|X=0}=P$ and $P_{Y|X=1}=Q$. Then,
its contraction coefficient $\etaKL(P_{Y|X})=\etachi(P_{Y|X})\eqdef \eta(\{P,Q\})$ satisfies
\label{thm:etaPQ}
\begin{equation}
\frac{H^2(P,Q)}{2} \leq \eta(\{P,Q\}) \leq H^2(P,Q) - \frac{H^4(P,Q)}{2}\,,
	\label{eq:etaPQ}
\end{equation}
where Hellinger distance is defined as $H^2(P,Q) \eqdef 2 - 2\int \sqrt{dP dQ}$.
\end{theorem}

\begin{remark}
An obvious upper bound is $\eta(\{P,Q\}) \leq \TV(P,Q)$ by \prettyref{thm:eta_ub}, which is worse
than \prettyref{thm:etaPQ} since $\TV$ is small than the square-root of the right-hand side of \prettyref{eq:etaPQ}. In
fact it is straightforward to verify that the upper bound holds with equality when the output $Y$ is also binary-valued.
In particular, \prettyref{thm:etaPQ} implies that $\eta(\{P,Q\})$ is always within \emph{a factor of two} of $H^2(P,Q)$.
\end{remark}

%
%
%

\begin{proof}
	
First notice the identities: 
\begin{align}
\chi^2(\Bern(\alpha)\|\Bern(\beta)) = & ~ \frac{(\alpha-\beta)^2}{\beta\bar\beta},	\nonumber \\
\chi^2(\alpha P + \bar \alpha Q \| \beta P + \bar\beta Q) = & ~ (\alpha-\beta)^2   \int \frac{(P-Q)^2}{\beta P +
\bar\beta Q}\,,	\nonumber
\end{align}
where we denote $\bar\alpha=1-\alpha$.
Therefore the (input-dependent) $\chi^2$-contraction coefficient is given by 
\[
\etachi(\Bern(\beta),P_{Y|X}) = \sup_{\alpha \neq \beta} \frac{\chi^2(\alpha P + \bar \alpha Q \| \beta P + \bar\beta Q)}{\chi^2(\Bern(\alpha)\|\Bern(\beta))}
= \beta\bar\beta  \int \frac{(P-Q)^2}{\beta P + \bar\beta Q} \triangleq \LC_{\beta}(P\|Q),
\]
where $\LC_{\beta}(P\|Q)$, clearly an $f$-divergence, is known as the Le Cam divergence (see, \eg, \cite[p.~889]{Vajda09}).
In view of \prettyref{thm:chi2KL}, the input-independent KL-contraction coefficient coincides with that of $\chi^2$ and hence
\[
\eta(\{P,Q\}) = \sup_{\beta \in (0,1)} \LC_{\beta}(P\|Q).
\]
Thus the desired bound \prettyref{eq:etaPQ} follows from the characterization of the joint range between pairs of $f$-divergence \cite{HV11}, namely, $H^2$ versus $\LC_\beta$, by taking the convex hull of their joint range restricted to Bernoulli distributions. Instead of invoking this general result, next we prove \prettyref{eq:etaPQ} using elementary arguments.
Since $\LC_{1/2}(P\|Q) = 1 - 2 \int \frac{dP dQ}{dP+dQ} \geq 1 - \int \sqrt{dP dQ} = \frac{1}{2}H^2(P,Q)$, the left inequality of \prettyref{eq:etaPQ} follows immediately. 
To prove the right inequality, by Cauchy-Schwartz, note that we have
$(1 - \frac{1}{2} H^2(P,Q))^2 = (\int\sqrt{dPdQ})^2 = (\int\sqrt{\beta dP+ \bar \beta dQ} \sqrt{\frac{dPdQ}{\beta dP+ \bar \beta dQ}} )^2 \leq \int \frac{dPdQ}{\beta dP+ \bar \beta dQ} = 1 - \LC_\beta(P\|Q)$, for any $\beta \in (0,1)$.
\end{proof}

\section{Simultaneously maximal couplings}
	\label{app:tv}

\begin{lemma}
		\label{lmm:tv}
Let $\calX$ and $\calY$ be Polish spaces. 
Given any pair of Borel probability measures $P_{XY},Q_{XY}$ on $\calX \times \calY$, there exists a coupling $\pi$ of $P_{XY}$ and $Q_{XY}$, namely, a joint distribution of $(X,Y,X',Y')$ such that $\calL(X,Y)=P_{XY}$ and $\calL(X',Y')=Q_{XY}$ under $\pi$, such that
\begin{equation}\label{eq:tve}
\pi\{(X,Y)\neq (X',Y')\} = \TV(P_{XY},Q_{XY}) \quad \text{and} \quad \pi\{ X \neq X'\} = \TV(P_{X},Q_{X}).
\end{equation}
\end{lemma}
\begin{remark} After submitting this manuscript, we were informed that this result is the main content of~\cite{SG79}.
For interested reader we keep our original proof which is different from~\cite{SG79} by relying on Kantorovich's dual
representation and, thus, is non-constructive.
\end{remark}
\begin{remark}
	\label{rmk:triple}
A triply-optimal coupling achieving in addition to~\eqref{eq:tve} also $\pi[Y\neq Y'] = \TV(P_Y,Q_Y)$ need not exist.
Indeed, consider the example where $X,Y$ are $\{0,1\}$-valued and
$$ P_{XY} = \begin{pmatrix} {1\over2} & 0 \\ 0 & {1\over2} \end{pmatrix}, 
 \quad 
  Q_{XY} = \begin{pmatrix} 0 & {1\over2} \\ {1\over2} & 0 \end{pmatrix}\,. $$
In other words, $X,Y\sim \Bern(1/2)$ under both $P$ and $Q$; however, $X = Y$ under $P$ and $X = 1-Y$ under $Q$.
Furthermore, 
since $\TV(P_X,Q_X)=\TV(P_Y,Q_Y)=0$, under any coupling $\pi_{XYX'Y'}$ of $P_{XY}$ and $Q_{XY}$
 that simultaneously couples $P_X$ to $Q_X$ and $P_Y$ to $Q_Y$ maximally, we have 
$X=X'$ and $Y=Y'$, which contradicts $X=Y$ and $X'=1-Y'$. 
On the other hand, it is clear that a doubly-optimal coupling (as
claimed by Lemma~\ref{lmm:tv}) exists: just take $X=X'=Y \sim \mathrm{Bern}(1/2)$ and $Y'=1-X'$. It is not hard to show that such a coupling also attains the minimum 
$$ \min_{\pi} \pi[(X,Y)\neq (X',Y')] + \pi[X\neq X'] + \pi[Y\neq Y'] = 2. $$
\end{remark}

\begin{proof}
	Define the cost function $c(x,y,x',y') \triangleq \indc{(x,y)\neq (x',y')}+\indc{x\neq x'} = 2 \indc{x\neq x'} + \indc{x = x', y\neq y'}$. Since the indicator of any open set is lower semicontinuous, so is $(x,y,x',y') \mapsto c(x,y,x',y')$. Applying Kantorovich's duality theorem (see, \eg, \cite[Theorem 1.3]{villani.topics}), we have
\begin{equation}
\min_{\pi \in \Pi(P_{XY},Q_{XY})} \Expect_\pi c(X,Y,X',Y')  = \max_{f,g}   \Expect_P[f(X,Y)]-\Expect_Q[g(X,Y)].	
	\label{eq:dual}
\end{equation}
where $f \in L_1(P),g\in L_1(Q)$ and 
\begin{equation}
	f(x,y) - g(x',y') \leq c(x,y,x',y').
	\label{eq:kanto}
\end{equation} 
Since the cost function is bounded, namely, $c$ takes values in $[0,2]$, applying \cite[Remark 1.3]{villani.topics}, we conclude that it suffices to consider $0 \leq f,g \leq 2$. Note that constraint \prettyref{eq:kanto} is equivalent to
\begin{align}
f(x,y) - g(x',y') & \leq 2, \forall x\neq x',\forall 	y\neq y' 	\nonumber \\
f(x,y) - g(x,y') & \leq 1, \forall x, \forall y \neq y'	\nonumber \\
f(x,y) - g(x,y) & \leq 0, \forall x, \forall y 	\nonumber 
\end{align}
where the first condition is redundant given the range of $f,g$. In summary, the maximum on the right-hand side of \prettyref{eq:dual} can be taken
over all $f,g$ satisfying the following constraints:
\begin{align}
0 \leq f, g & \leq 2 \nonumber \\
f(x,y) - g(x,y') & \leq 1, \forall x,  y \neq y'	\nonumber \\
f(x,y) - g(x,y) & \leq 0, \forall x, y 	\nonumber 
\end{align}
Then
\begin{align}
\max_{f,g}   \Expect_P[f(X,Y)]-\Expect_Q[g(X,Y)] 
= & ~ \int_{\calX} \max_{\phi,\psi} \sth{  \int_\calY   p(x,y) \phi(y) -  q(x,y) \psi(y) } \label{eq:dual1}
\end{align}
where the maximum on the right-hand side is over $\phi,\psi: \calY \to \reals$ satisfying
\begin{equation}
\begin{aligned}
0 \leq \phi, \psi \leq 2 \\
\phi(y) -  \psi(y') \leq 1, \forall y \neq y'	\\
\phi(y) -  \psi(y) \leq 0, \forall y 	
\end{aligned}	
	\label{eq:phipsi}
\end{equation}

The optimization problem in the bracket on the RHS of \prettyref{eq:dual1} can be solved using the following lemma:
\begin{lemma}
		\label{lmm:yp}
Let $p,q \geq 0$. Let $(x)_+ \triangleq \max\{x,0\}$. Then
\begin{equation}
	\max_{\phi,\psi} \sth{  \int_\calY   p \phi -  q \psi : 0 \leq \phi \leq \psi \leq 2, \sup \phi \leq 1 + \inf \psi }
	= \int   (p-q)_+  + \pth{\int   (p-q)}_+.
	\label{eq:yp}
\end{equation}
\end{lemma}

\begin{proof}
	First we show that it suffices to consider $\phi=\psi$. Given any feasible pair $(\phi,\psi)$, set $\phi'= \max\{\phi, \inf \psi\}$. To check that $(\phi',\phi')$ is a feasible pair, note that clearly $\phi'$ takes values in $[0,2]$. Furthermore, $\sup \phi' \leq \sup \phi \leq 1 + \inf \psi \leq 1 + \inf \phi'$. Therefore the maximum on the left-hand side of \prettyref{eq:yp} is equal to
	\[
	\max_{\phi} \sth{  \int_\calY   (p -  q) \phi :  0 \leq \phi \leq 2, \sup \phi \leq 1 + \inf \phi }.
	\]
	Let $a = \inf \phi$. Then
	\begin{align*}
	 \max_{\phi} \sth{  \int   (p-q) \phi :  0 \leq \phi \leq 2, \sup \phi \leq 1 + \inf \phi }	
= & ~ \sup_{0 \leq a \leq 2} \max_{\phi} \sth{  \int   (p-q) \phi :  a \leq \phi \leq 2 \wedge (1+a) }	\\
= & ~ \sup_{0 \leq a \leq 1} \max_{\phi} \sth{  \int   (p-q) \phi :  a \leq \phi \leq 1+a }	\\
= & ~ \sup_{0 \leq a \leq 1}  \sth{  (1+a) \int   (p-q)_+  + a \int   (p-q)_-  }	\\
= & ~ \sup_{0 \leq a \leq 1}  \sth{   \int   (p-q)_+  + a \int   (p-q)  }	\\
= & ~	 \int   (p-q)_+  + \pth{\int   (p-q)}_+. 
\end{align*}
\end{proof}

Applying \prettyref{lmm:yp} to \prettyref{eq:dual1} for fixed $x$, we have
\begin{align*}
& ~ \max_{f,g}   \Expect_P[f(X,Y)]-\Expect_Q[g(X,Y)] \\
= & ~ \int_\calX\pth{\int_\calY   (p(x,y)-q(x,y))_+  + (p(x)-q(x))_+} \\
= & ~ \int_\calX \int_\calY   (p(x,y)-q(x,y))_+  + \int_\calX (p(x)-q(x))_+ =  \TV(P_{XY},Q_{XY}) + \TV(P_{X},Q_{X})  \label{eq:dual2}
\end{align*}
Combining the above with \prettyref{eq:dual}, we have
\begin{equation*}
\min_{\pi_{XYX'Y'}} \pi\{(X,Y)\neq (X',Y')\} + \pi\{ X \neq X'\}  = \TV(P_{XY},Q_{XY}) + \TV(P_{X},Q_{X}).
\end{equation*}
Since $\pi\{(X,Y)\neq (X',Y')\} \geq \TV(P_{XY},Q_{XY})$ and $\pi\{ X \neq X'\} \geq \TV(P_{X},Q_{X})$ for any $\pi$, the minimizer
of the sum on the left-hand side achieves equality simultaneously for both terms, proving the theorem.
\end{proof}

%

%
%
%
%
%
%
%
%
%
%


\begin{thebibliography}{DMLM03}

\bibitem[AG76]{AG76}
R.~Ahlswede and P.~G{\'a}cs.
\newblock Spreading of sets in product spaces and hypercontraction of the
  {M}arkov operator.
\newblock {\em Ann. Probab.}, pages 925--939, 1976.

\bibitem[AGKN13]{AGKN13}
Venkat Anantharam, Amin Gohari, Sudeep Kamath, and Chandra Nair.
\newblock On maximal correlation, hypercontractivity, and the data processing
  inequality studied by {E}rkip and {C}over.
\newblock {\em arXiv preprint arXiv:1304.6133}, 2013.

\bibitem[Ash65]{ash-itbook}
Robert~B. Ash.
\newblock {\em Information Theory}.
\newblock Dover Publications Inc., New York, NY, 1965.

\bibitem[Bir57]{GB57}
G.~Birkhoff.
\newblock Extensions of {J}entzsch's theorem.
\newblock {\em Trans. of AMS}, 85:219--227, 1957.

\bibitem[BK98]{BK98-infer}
Xavier Boyen and Daphne Koller.
\newblock Tractable inference for complex stochastic processes.
\newblock In {\em Proceedings of the 14th Conference on Uncertainty in
  Artificial Intelligence---UAI 1998}, pages 33--42. San Francisco: Morgan
  Kaufmann, 1998.
\newblock Available at \url{http://www.cs.stanford.edu/~xb/uai98/}.

\bibitem[CIR{\etalchar{+}}93]{CIR93}
J.E. Cohen, Yoh Iwasa, Gh. Rautu, M.B. Ruskai, E.~Seneta, and Gh. Zbaganu.
\newblock Relative entropy under mappings by stochastic matrices.
\newblock {\em Linear algebra and its applications}, 179:211--235, 1993.

\bibitem[CK81]{CK}
I.~Csisz\'{a}r and J.~K\"{o}rner.
\newblock {\em Information Theory: Coding Theorems for Discrete Memoryless
  Systems}.
\newblock Academic, New York, 1981.

\bibitem[CKZ98]{CKZ98}
J.~E. Cohen, J.~H.~B. Kempermann, and Gh. Zb\u{a}ganu.
\newblock {\em Comparisons of Stochastic Matrices with Applications in
  Information Theory, Statistics, Economics and Population}.
\newblock Springer, 1998.

\bibitem[Cou12]{TC12}
T.~Courtade.
\newblock {\em Two Problems in Multiterminal Information Theory}.
\newblock PhD thesis, U. of California, Los Angeles, CA, 2012.

\bibitem[CPW15]{CPW15-journal}
F.~Calmon, Y.~Polyanskiy, and Y.~Wu.
\newblock Strong data processing inequalities for input-constrained additive
  noise channels.
\newblock {\em arXiv}, December 2015.
\newblock arXiv:1512.06429.

\bibitem[CRS94]{CRS94}
M.~Choi, M.B. Ruskai, and E.~Seneta.
\newblock Equivalence of certain entropy contraction coefficients.
\newblock {\em Linear algebra and its applications}, 208:29--36, 1994.

\bibitem[Csi67]{IC67}
I.~Csisz\'ar.
\newblock Information-type measures of difference of probability distributions
  and indirect observation.
\newblock {\em Studia Sci. Math. Hungar.}, 2:229--318, 1967.

\bibitem[Daw75]{dawson1975information}
DA~Dawson.
\newblock Information flow in graphs.
\newblock {\em Stoch. Proc. Appl.}, 3(2):137--151, 1975.

\bibitem[DJW13]{duchi2013local}
John~C Duchi, Michael Jordan, and Martin~J Wainwright.
\newblock Local privacy and statistical minimax rates.
\newblock In {\em Foundations of Computer Science (FOCS), 2013 IEEE 54th Annual
  Symposium on}, pages 429--438. IEEE, 2013.

\bibitem[DMLM03]{MLM03}
P.~Del~Moral, M.~Ledoux, and L.~Miclo.
\newblock On contraction properties of {M}arkov kernels.
\newblock {\em Probab. Theory Relat. Fields}, 126:395--420, 2003.

\bibitem[Dob56]{RLD56}
R.~L. Dobrushin.
\newblock Central limit theorem for nonstationary {M}arkov chains. {I}.
\newblock {\em Theory Probab. Appl.}, 1(1):65--80, 1956.

\bibitem[Dob70]{RLD70}
R.~L. Dobrushin.
\newblock Definition of random variables by conditional distributions.
\newblock {\em Theor. Probability Appl.}, 15(3):469--497, 1970.

\bibitem[Doe37]{doeblin1937cas}
Wolfgang Doeblin.
\newblock {\em Le cas discontinu des probabilit{\'e}s en cha{\^\i}ne}.
\newblock na, 1937.

\bibitem[EC98]{EC98}
Elza Erkip and Thomas~M. Cover.
\newblock The efficiency of investment information.
\newblock {\em {IEEE} Trans. Inf. Theory}, 44(3):1026--1040, 1998.

\bibitem[EGK11]{ELGK11}
Abbas El~Gamal and Young-Han Kim.
\newblock {\em Network information theory}.
\newblock Cambridge university press, 2011.

\bibitem[EKPS00]{EKPS00}
William Evans, Claire Kenyon, Yuval Peres, and Leonard~J Schulman.
\newblock Broadcasting on trees and the {I}sing model.
\newblock {\em Ann. Appl. Probab.}, 10(2):410--433, 2000.

\bibitem[ES99]{evans1999signal}
William~S Evans and Leonard~J Schulman.
\newblock Signal propagation and noisy circuits.
\newblock {\em {IEEE} Trans. Inf. Theory}, 45(7):2367--2373, 1999.

\bibitem[Gol79]{SG79}
Sheldon Goldstein.
\newblock Maximal coupling.
\newblock {\em Probability Theory and Related Fields}, 46(2):193--204, 1979.

\bibitem[HV11]{HV11}
P.~Harremo{\"e}s and I.~Vajda.
\newblock On pairs of $f$-divergences and their joint range.
\newblock {\em {IEEE} Trans. Inf. Theory}, 57(6):3230--3235, Jun. 2011.

\bibitem[Lau96]{lau96}
Steffen~L Lauritzen.
\newblock {\em {G}raphical {M}odels}.
\newblock Oxford University Press, 1996.

\bibitem[LCV15]{liu2015secret}
Jingbo Liu, Paul Cuff, and Sergio Verdu.
\newblock Secret key generation with one communicator and a zero-rate one-shot
  via hypercontractivity.
\newblock {\em arXiv preprint arXiv:1504.05526}, 2015.

\bibitem[Led99]{MLxx}
M.~Ledoux.
\newblock Concentration of measure and logarithmic {S}obolev inequalities.
\newblock {\em Seminaire de probabilites XXXIII}, pages 120--216, 1999.

\bibitem[Mar06]{markov1906extension}
Andrey~Andreyevich Markov.
\newblock Extension of the law of large numbers to dependent quantities.
\newblock {\em Izv. Fiz.-Matem. Obsch. Kazan Univ.(2nd Ser)}, 15:135--156,
  1906.

\bibitem[MS77]{macwilliams1977theory}
Florence~Jessie MacWilliams and Neil James~Alexander Sloane.
\newblock {\em The theory of error correcting codes}.
\newblock Elsevier, 1977.

\bibitem[MZ15]{makur2015bounds}
Anuran Makur and Lizhong Zheng.
\newblock Bounds between contraction coefficients.
\newblock {\em arXiv preprint arXiv:1510.01844}, 2015.

\bibitem[Nai14]{Nair2015}
C.~Nair.
\newblock Equivalent formulations of hypercontractivity using information
  measures.
\newblock In {\em Proc. 2014 Zurich Seminar on Comm.}, 2014.

\bibitem[Ord16]{OO16-isit}
Or~Ordentlich.
\newblock Novel lower bounds on the entropy rate of binary hidden {M}arkov
  processes.
\newblock In {\em Proc. 2016 IEEE Int. Symp. Inf. Theory (ISIT)}, Barcelona,
  Spain, July 2016.

\bibitem[PW15]{PWLN}
Y.~Polyanskiy and Y.~Wu.
\newblock Lecture notes on information theory.
\newblock 2016.
\newblock http://people.lids.mit.edu/yp/homepage/data/itlectures\_v4.pdf

\bibitem[PW16]{PW14a}
Yury Polyanskiy and Yihong Wu.
\newblock Dissipation of information in channels with input constraints.
\newblock {\em {IEEE} Trans. Inf. Theory}, 62(1):35--55, January 2016.
\newblock also arXiv:1405.3629.

\bibitem[Rag13]{Raginsky13}
Maxim Raginsky.
\newblock Logarithmic {S}obolev inequalities and strong data processing
  theorems for discrete channels.
\newblock In {\em 2013 IEEE International Symposium on Information Theory
  Proceedings (ISIT)}, pages 419--423, 2013.

\bibitem[Rag14]{Raginsky14}
Maxim Raginsky.
\newblock Strong data processing inequalities and $\phi$-sobolev inequalities
  for discrete channels.
\newblock {\em arXiv preprint arXiv:1411.3575}, November 2014.

\bibitem[Sam15]{AS15-noisyfunc}
Alex Samorodnitsky.
\newblock On the entropy of a noisy function.
\newblock {\em arXiv preprint arXiv:1508.01464}, August 2015.

\bibitem[Sar58]{OS58}
O.~V. Sarmanov.
\newblock Maximal correlation coefficient (non-symmetric case).
\newblock {\em Dokl. Akad. Nauk SSSR}, 121(1):52--55, 1958.

\bibitem[Vaj09]{Vajda09}
I.~Vajda.
\newblock On metric divergences of probability measures.
\newblock {\em Kybernetika}, 45(6):885--900, 2009.

\bibitem[Vil03]{villani.topics}
C.~Villani.
\newblock {\em {Topics in optimal transportation}}.
\newblock American Mathematical Society, Providence, RI, 2003.

\bibitem[XR15]{xu2015converses}
Aolin Xu and Maxim Raginsky.
\newblock Converses for distributed estimation via strong data processing
  inequalities.
\newblock In {\em Proc. 2015 IEEE Int. Symp. Inf. Theory (ISIT)}, Hong Kong,
  CN, July 2015.

\bibitem[ZY97]{zhang1997non}
Zhen Zhang and Raymond~W Yeung.
\newblock A non-{S}hannon-type conditional inequality of information
  quantities.
\newblock {\em {IEEE} Trans. Inf. Theory}, 43(6):1982--1986, 1997.

\end{thebibliography}

\newcommand{\etalchar}[1]{$^{#1}$}

\end{document}